\documentclass[a4paper, 12pt]{article}
\usepackage{amsfonts, amsmath, amssymb, epsfig, subfigure, graphicx}

\usepackage[toc,page]{appendix}
\usepackage{authblk}
\usepackage{float}
\usepackage{natbib}

\usepackage{tikz}
\usepackage{tabularx,ragged2e,booktabs,caption}
\usetikzlibrary{calc,patterns}
\usepackage{xcolor}
\usepackage{lineno}
\usepackage{amsthm}
\usepackage{lipsum}
\usepackage{multirow}
\newtheorem{proposition}{Proposition}
\newtheorem{prop}{Proposition}

\renewenvironment{proof} 
{\noindent {\bf Proof }}
{{\hfill $\Box$}\\
}

\newcommand{\commentout}[1]{}
 \footnotesep=\baselineskip

\usepackage{fancyhdr}
\begin{document}


\chead{}

\author{H\'el\`ene Barcelo$^1$ and Valerio Capraro$^2$}
\affil{$^1$Mathematical Sciences Research Institute, Berkeley, USA.\\ $^2$Middlesex University Business School, London, UK.}
\title{On the heterogeneity of people's (dis)honest behavior along two dimensions: Cost of lying and cost of acquiring information}
\maketitle

\begin{abstract}
This paper studies lying in a novel context. Previous work has focused on situations in which people are either fully aware of the economic consequences of all available actions (e.g., die-under-cup paradigm), or they are uncertain, but this uncertainty cannot be cleared in any way (e.g., sender-receiver game). On the contrary, in reality, people oftentimes know that they will have a chance to lie, they are initially uncertain about the economic consequences of the available actions, but they can invest resources (e.g., time) to find them out. Here we capture the essence of this type of situations by means of a novel decision problem. Two experiments provide evidence of four empirical regularities regarding the distribution of choices, and suggest that participants vary along two dimensions: the moral cost of lying, and the cost of investing time to find out the payoffs associated to the available actions. Taking inspiration from these observations, we introduce a model that is consistent with the main empirical results.
    
\end{abstract}

\emph{Journal of Economic Literature classification codes:} C70, C92, D03.


\pagebreak

\section{Introduction}

The conflict between honesty and dishonesty is at the core of all economic and social interactions that involve communication with asymmetric information. In these situations, people have the opportunity to misreport their private pieces of information. Dishonesty clearly has negative impact on government and companies. For example, every year, tax evasion costs about \$100 billion to the U.S. government \citep{gravelle2009tax}, and insurance fraud costs more than \$40 billion to insurance companies\footnote{See \texttt{https://www.fbi.gov/stats-services/publications/insurance-fraud/insurance\_fraud}}. 

In the past decade, economists and psychologists have started studying (dis)honesty using incentivized economic problems \citep{erat2012white,fischbacher2013lies,gneezy2005deception,hurkens2009would,kartik2009strategic,levine2014liars,levine2015prosocial,mazar2008dishonesty,sheremeta2013liars,weisel2015collaborative,wiltermuth2011cheating}. For example, they have explored the effect on honesty of many exogenous and endogenous variables, such as: demographic characteristics \citep{abeler2018preferences,biziou2015does,cappelen2013we,capraro2018gender,childs2012gender,dreber2008gender,friesen2012individual,erat2012white}; social and moral preferences \citep{biziou2015does,levine2014liars,levine2015prosocial,shalvi2014oxytocin}; incentives \citep{dreber2008gender,erat2012white,ezquerra2018gender,fischbacher2013lies,gneezy2005deception,gneezy2018lying,kajackaite2017incentives,mazar2008dishonesty}; and cognitive mode \citep{andersen2018allowing,cappelen2013we,capraro2017does,capraro2019time,gino2011unable,gunia2012contemplation,lohse2018deception,shalvi2012honesty}.

While differing in important details, all these experiments share a common property: they are \emph{static}, in the sense that they focus either on situations in which people are fully aware of the consequences of all available options, or on situations for which there is uncertainty about the consequences of the options, but this uncertainty cannot be cleared in any way (that is, there is no room for learning the consequences of the available actions). To be more precise, one set of studies (e.g., \cite{fischbacher2013lies}) implements a die-under-cup paradigm, in which participants roll a die privately (under a cup), and then are asked to report the outcome, knowing that they would be paid according to the reported number. Since participants know the outcome of the die and know the payoff function, they are fully aware of the consequences of all available actions. The other set (e.g., \cite{gneezy2005deception}), instead, implements a sender-receiver game, in which a person, named \emph{sender}, is given a private information and is asked to report it to another person, the \emph{receiver}, whose role is to guess the original piece of information given to the sender. The payoffs of both the sender and the receiver depend on whether the receiver guesses the original piece of information. Since the sender has no way to know whether or not the receiver is going to believe his message, the sender has no way to clear the uncertainty about the economic consequences of his available actions\footnote{There are also studies implementing a sender-receiver game in which the receiver makes no choice (e.g., \cite{biziou2015does,capraro2017does,capraro2019time}). In this case, there is no uncertainty at all and, consequently, these studies belong to the previous class of studies, in which subjects are fully aware of the consequences of the available actions. There are also studies in which subjects perform a task and then are paid according to the self-reported score in the task \citep{mazar2008dishonesty,gino2013self}. In this case, either there is no uncertainty regarding the payoffs associated to the available actions or, if there is an uncertainty, for example regarding one's own real score, this uncertainty cannot be clear. Therefore, also these studies belong to the previously mentioned classes.}.

Therefore, although important, these studies provide an incomplete picture as, in reality, people oftentimes make decisions within a \emph{dynamic} setting in which they know that they will have a chance to tell the truth or lie, but they do not initially know the exact material consequences of these actions. They have to invest resources (e.g., time) to find them out. 

Many real situations are, at their core, dynamic in this sense. For example, before starting their tax declaration, people know that they are going to have the chance to lie (or misrepresent some facts of their declaration) in order to pay less taxes; however, they do not initially know how much money they would save for each possible lie they could tell. 

In this work, we capture the essence of this type of situations by means of a novel decision problem, in which participants, after being informed that they will have a chance to report a piece of information either truthfully or not (several different ways to misreport the piece of information will be possible), and before actually choosing how to report that piece of information, need to invest time to explore the payoffs associated to the available strategies.

This decision problem turns out to be particularly interesting also because it opens the way to a completely new set of empirical questions. Indeed, most previous experimental studies on deception implement either binary decision problems --- in which subjects get to choose between two strategies, one corresponding to telling the truth and one corresponding to lying \citep{biziou2015does,capraro2017does,gunia2012contemplation} --- or decision problems in which there are many different ways to lie, but finding which lie maximizes the payoff is trivial \citep{cappelen2013we,erat2012white,gneezy2005deception,shalvi2012honesty,fischbacher2013lies,sheremeta2013liars}. On the contrary, in our case, decision makers can lie in several different ways: they can explore all potential payoffs and then report the piece of information that corresponds to the global maximum; or they can explore only some payoffs and then report a piece of information that corresponds to an early local maximum (a local maximum that can be found easily and quickly --- to be properly defined later); or they can explore the payoff corresponding to reporting the truth and then decide to lie only if this payoff is low (in case they decide to lie, they can do it in at least two ways: they can choose the global maximum or they ``can stretch the truth'', by looking around the truth in search of a profitable deviation); subjects can even be indifferent and pick one choice at random.

To shed light on how people make decisions in this context, we conduct
two experiments, Study 1 and Study 2 (see Section \ref{se:study1}, Section \ref{se:limitation}, and Section \ref{se:study2}). These experiments provide evidence of four major results: (i) Very few participants stretch the truth; (ii) There is a significant proportion of participants that report the truth without first finding out its corresponding payoff; (iii) Among participants who decide to find out the payoff corresponding to telling the truth, this payoff is positively correlated with honesty; (iv) If there is a very high payoff (but smaller than the global maximum) that participants can find easily and quickly, then a significant proportion of participants report the piece of information corresponding to this very high payoff. 

These results suggest that participants vary in (at least) two dimensions: the moral cost of lying, and the cost of finding out payoffs. In Section \ref{se:model} we follow this idea and we introduce a model. According to this model, participants sequentially choose whether to report a piece of information at time $t$ or to pay a cost to reveal one payoff at time $t+1$. The utility is assumed to depend on the payoff corresponding to the reported information, the moral cost of lying (that is paid only if a participant decides to lie), and the cumulative cost of finding out payoffs. We show that this model is consistent with the four empirical regularities listed above.

Section \ref{se:discussion} concludes.

\section{Overview of the experimental design}\label{se:overview}

Our goal is to build a decision problem that allows us to study how people make decisions in situations in which they know that they will have a chance to lie, but they do not initially know the economic consequences of the available actions. Perhaps the simplest way to do this is as follows. Participants are given a list of payoffs that are initially covered. Then they are given, as private information, a ``position''. Finally, their are told that their job is to report the given position, and that they will be paid the amount of money corresponding to the position they report. In this way, participants are incentivized to look for the position corresponding to the maximum payoff and report that position. (Note that it is important to ask participants to report the position, and not the payoff in that position, otherwise it would not be possible to make ``blind'' decisions, that is, decisions without knowing the corresponding payoff). At this point, participants can either report a position without finding out any payoff, or can invest time to sequentially uncover some or all payoffs. This design certainly has a number of positive features. For example, it would allow us to see the number of payoffs uncovered by the participants before reporting a position. This, in particular, would allow us to know with certainty whether a participant reports a position with or without knowing the corresponding payoff. Yet, this design also has an important drawback. Previous work suggests that actively finding out the economic consequences of the available actions is perceived, by observers, as a signal of selfishness \citep{capraro2016know,jordan2016uncalculating}. This implies that pro-social participants might be particularly reluctant to actively uncover the payoffs. To avoid this potential problem, we opted for keeping the same decision problem structure, but, instead of giving participants the list of payoffs covered, we give it to them already uncovered. Note that, in this case, participants still need to invest time to find the payoff corresponding to each position they might report, because they have to count up to the position corresponding to the desired payoff. Therefore, the key point of the decision problem remains intact. Moreover, we avoid the problem that pro-social participants might be particularly reluctant to explore the payoffs. But clearly this comes with a price: we will not be able to directly see whether participants make a decision with or without knowing the payoff corresponding to telling the truth. In Study 1, we address this point by putting the given position late in the list and by analyzing response times. The idea (supported by the results, as it will be discussed later) is that fast participants are less likely than slow participants to have found out the payoff corresponding to telling the truth. 

\section{Study 1}\label{se:study1}

\subsection{Procedure}

We recruited participants on the online platform Amazon Mechanical Turk (AMT). AMT experiments are easy and cheap because participants attend from their homes by simply completing an online incentivized survey that takes no more than a few minutes. This allows researchers to significantly decrease the stakes of the experiment, without compromising the results. Several studies have indeed shown that data gathered using AMT are of no less quality than data gathered using the standard physical lab \citep{arechar2018conducting,berinsky2012evaluating,goodman2013data,horton2011online,mason2012conducting,paolacci2010running,paolacci2014inside}. Moreover, as an upside with respect to standard laboratory experiments, AMT experiments use samples that are more heterogeneous than the standard laboratory experiments, that are typically conducted using a pool of students \citep{berinsky2012evaluating,paolacci2014inside}. 

After entering their Turk ID, participants were told that they will generate two pieces of information:
\begin{itemize}
    \item A position, which will be a number, for example 14. 
    \item 	A list of potential payoffs between 1 and 90 cents, such as: 
    \end{itemize}
    $$ 23\quad 73\quad  34\quad  22\quad  2\quad  11\quad  54\quad  21\quad  44\quad  3\quad  22\quad  6\quad  89\quad  45\quad  67\quad  23\quad  65\quad  46\quad  77\quad  1\quad  86\quad  5\quad  4$$

All the payoffs were shown on the same line, divided by two spaces.

Participants were told that their payoff will be equal to the number in the given position. Thus, in this case, it would be 45 cents, because 45 is in position 14. (We chose this particular example, because the payoff is essentially equal to the expected payoff from a random sequence, which is 45.5 cents --- we wanted to avoid priming participants towards being lucky or unlucky).

Then we told participants that their job will be to report the position they generated. We made it clear that the survey was anonymous and that, at the moment in which participants report their position, we do not know the position they generated.\footnote{This is not deception: participants were randomly assigned to different positions by a computer program and we could find out the position corresponding to a given participant only at the end of the survey.} After these general instructions, participants answered six comprehension questions with mock positions and mock lists. The positions and the lengths of the mock lists were all different, in order to avoid that the participants get used to finding the payoff associated with telling the truth without actually counting, and to avoid that they count backwards. Only participants who answered correctly all questions were allowed to proceed to the real experiment. Comprehension questions included questions about payoff maximization. So, participants who passed this attention test were aware of the fact that they could increase their payoff by misreporting the position they generated.

After the comprehension questions, participants were told that the real experiment was about to start, they were reminded that they will generate a position and a list of potential payoffs, and that their payoff for the survey will be equal to the number in the position they report. Then we asked them to press the next button to start playing.

In the next screen, participants were randomly divided in two treatments. In the Lucky condition, they were communicated that the position they generated was 22; in the Unlucky condition, they were communicated that the position they generated was 19.\footnote{The reason why the first condition is called Lucky while the second one is called Unlucky will be clear in the next paragraph. Clearly, the words Lucky and Unlucky were not used in the instructions: participants were only communicated their position.} Participants were asked to take note of this position on a piece of paper, and then to press the next button to generate the list of potential payoffs.\footnote{We chose to ask participants to take note of their position on a piece of paper in order to avoid that participants forget their position and end up lying not because they want to, but because they do not remember their position. We were particularly concerned with this point because participants in the next screen will face the list of payoffs and therefore might confuse the position (which is a number) with the payoffs (other numbers).} 

In the following screen, all participants were shown the same list of potential payoffs, which was: 

$$25\,\,\,\, 3\,\,\,\, 63\,\,\,\, 54\,\,\,\, 28\,\,\,\, 70\,\,\,\, 37\,\,\,\, 36\,\,\,\, 26\,\,\,\, 31\,\,\,\, 43\,\,\,\, 15\,\,\,\, 30\,\,\,\, 60\,\,\,\, 33\,\,\,\, 37\,\,\,\, 15\,\,\,\, 63\,\,\,\, 16\,\,\,\, 50\,\,\,\, 4\,\,\,\, 71\,\,\,\, 79\,\,\,\, 2\,\,\,\, 85\,\,\,\, 48$$

Then, participants were asked to report the position they generated.\footnote{We also conducted a Time Pressure condition, in which participants were asked to report their position within 15 seconds (median response time in the baseline). The reason why we also tested this condition is because there is a standing debate about whether honesty is intuitive or requires deliberation \citep{capraro2017does,capraro2019time,gunia2012contemplation,lohse2018deception,shalvi2012honesty}, and we wanted to see whether cognitive process would also matter in our case. It would not. We found that time pressure has no effect on honesty, although it has quite an obvious effect on the distribution of reported positions: it switches some would-be global maximizers into local maximizers, and it generates a small proportion of confused subjects. We report the analysis in the Appendix.} We took a measure of response time using a timer. The timer was not visible to participants.  

Before moving on, we make two observations about the list of potential payoffs. First, if participants in the Unlucky condition report the true position, then they get only 16 cents (because 16 is in position 19); if participants in the Lucky condition report the true position, then they get 71 cents (because 71 is in position 22), which is close to the maximum available payoff, which is 85 cents. This is why the former condition is named Unlucky, whereas the latter is named Lucky. Second, both Position 19 and Position 22 (the true positions) are adjacent to positions with a greater payoff. Thus, both lucky and unlucky participants can ``stretch the truth'' and get a greater payoff very easily (unlucky participants can increase their payoff by 47 cents or 34 cents, if they report Position 18 or Position 20, instead of their true position; similarly, lucky participants can increase their payoff very easily (by 8 cents) by reporting Position 23, instead of their true position). In other words: conditional on finding out the payoff corresponding to telling the truth, reporting the true position is essentially as easy as lying. Also this property is crucial: it allows us to avoid having subjects who, after finding out the payoff corresponding to telling the truth, prefer telling the truth over lying only because telling the truth is easier than lying.

After reporting the position, subjects were asked standard demographic questions (sex, age, education) and then they were given the completion code needed to submit the survey to AMT and claim for the payment.

We refer to the Appendix for full experimental instructions.

\subsection{Results}

\emph{Participants}

We recruited 400 participants located in the US. As it is usual in AMT experiments, in case of multiple observations (defined as those with either the same TurkID or the same IP address) we kept only the first observation (as determined by the starting date) and discarded the rest. Moreover, we excluded from the analysis participants who submitted the survey without completing it and participants who failed one or more comprehension questions. After this procedure, we were left with 347 valid observations (mean age = 36.8, females = 45.8\%). The proportion of participants that has been eliminated is in line with previous research using similar games \citep{horton2011online}.
\bigskip

\emph{Distribution of reported positions in Study 1.}

Figure 1 plots the histogram of the reported positions. Note that Position 19 and Position 22 correspond to the true positions in the Unlucky and Lucky conditions, respectively. Therefore, the majority of people acted honestly (overall rate of honesty = 84.1\%). This is in line with a recent meta-analysis, arguing that people lie surprisingly little in economic experiments \citep{abeler2018preferences}. Regarding the liars, observe that Position 25 is the one corresponding to the global maximum, while Position 3, Position 6, Position 18, and Position 23 correspond to the local maxima. It follows that virtually all liars maximized their payoff (either locally or globally). Moreover, among the 118 liars, only 10 chose to ``stretch the truth'', by choosing a position adjacent to the true position. In other words, 108 out of 118 liars (91.5\%), either reported the position corresponding to the global maximum, or reported a position corresponding to an early local maximum, where early local maximum stands for positions 3 and 6.
\bigskip

\textbf{Result 1.} Very few participants stretch the truth, that is, the vast majority of participants either report the true position or a position that is \emph{not} adjacent to the true position; in this latter case, they almost always report a position corresponding to a maximum (either global or local).
\bigskip

\begin{figure}
  \caption{Histogram of the positions reported in Study 1.}
  \centering
    \includegraphics[width=0.7\textwidth]{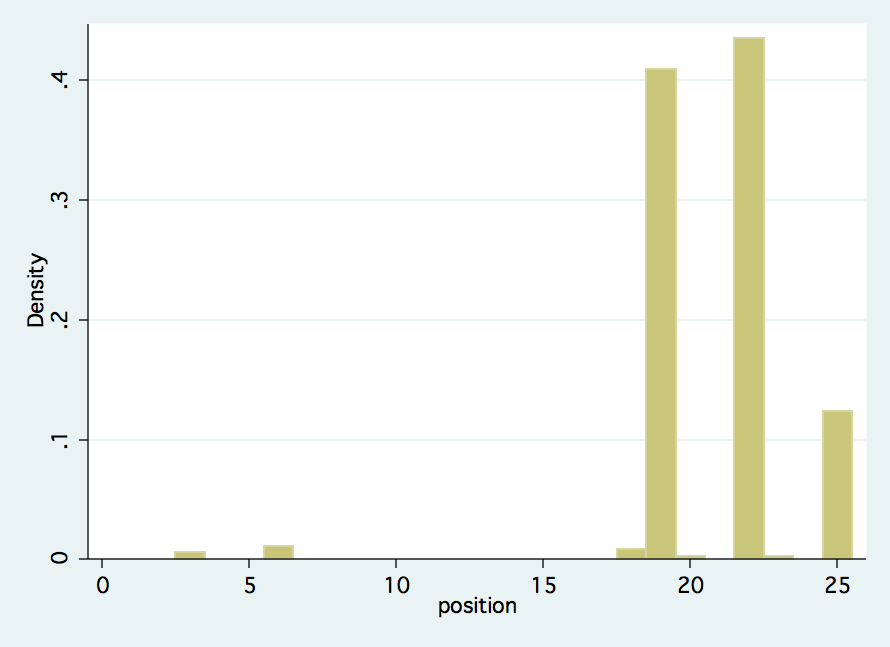}
\end{figure}

\bigskip

\emph{Heterogeneity in people's decision-making}

In this section, we explore whether participants find out the payoff corresponding to telling the truth before reporting a position, or whether there are some participants who report a position without first finding out the payoff corresponding to telling the truth. Apart from shedding light on how people make decisions in our experiment, this would also help us understand whether investing time to find out the payoffs has a cost for the participants.

To this end, we use response times. The idea is that participants who find out the payoff corresponding to telling the truth will take less time to report a position compared to participants who stop finding out the payoffs before reaching the payoff corresponding to telling the truth. Following this intuition, we divide subjects in two subsamples, by doing a median split on the response times: we define Fast (Slow) participants as those who take shorter (longer) than the median response time to report a position. Clearly, we do not expect this to be an exact measure: there will be Slow subjects who did not find out the payoff corresponding to telling the truth; conversely, there will be Fast subjects who found out the payoff corresponding to telling the truth. However, as we now show, this turns out to be a useful proxy. To see this, we start with the observation that, in line with previous literature, we expect that Unlucky subjects lie more than Lucky subjects, because they have a larger incentive to lie \citep{gneezy2005deception,gneezy2018lying}. We expect to see this effect also in our experiment, but only among subjects who found out the payoff corresponding to telling the truth (because the other ones do not know whether they have been lucky or not). Therefore, if, as we argue, most Fast participants did not find out the payoff corresponding to telling the truth and most Slow participants found out the payoff corresponding to telling the truth, then we would see that being Unlucky vs Lucky has an effect among Slow participants but not among Fast participants.

This prediction is indeed supported by the data. Figure 2 reports the rate of honesty split by whether they were Lucky or Unlucky and by response time (Fast participants vs Slow participants). 

\begin{figure}
  \caption{Rate of honesty among participants in Study 1, split by whether they were Lucky or Unlucky and by response time (faster half vs slower half). Error bars represent standard errors of the means.}
  \centering
    \includegraphics[width=0.7\textwidth]{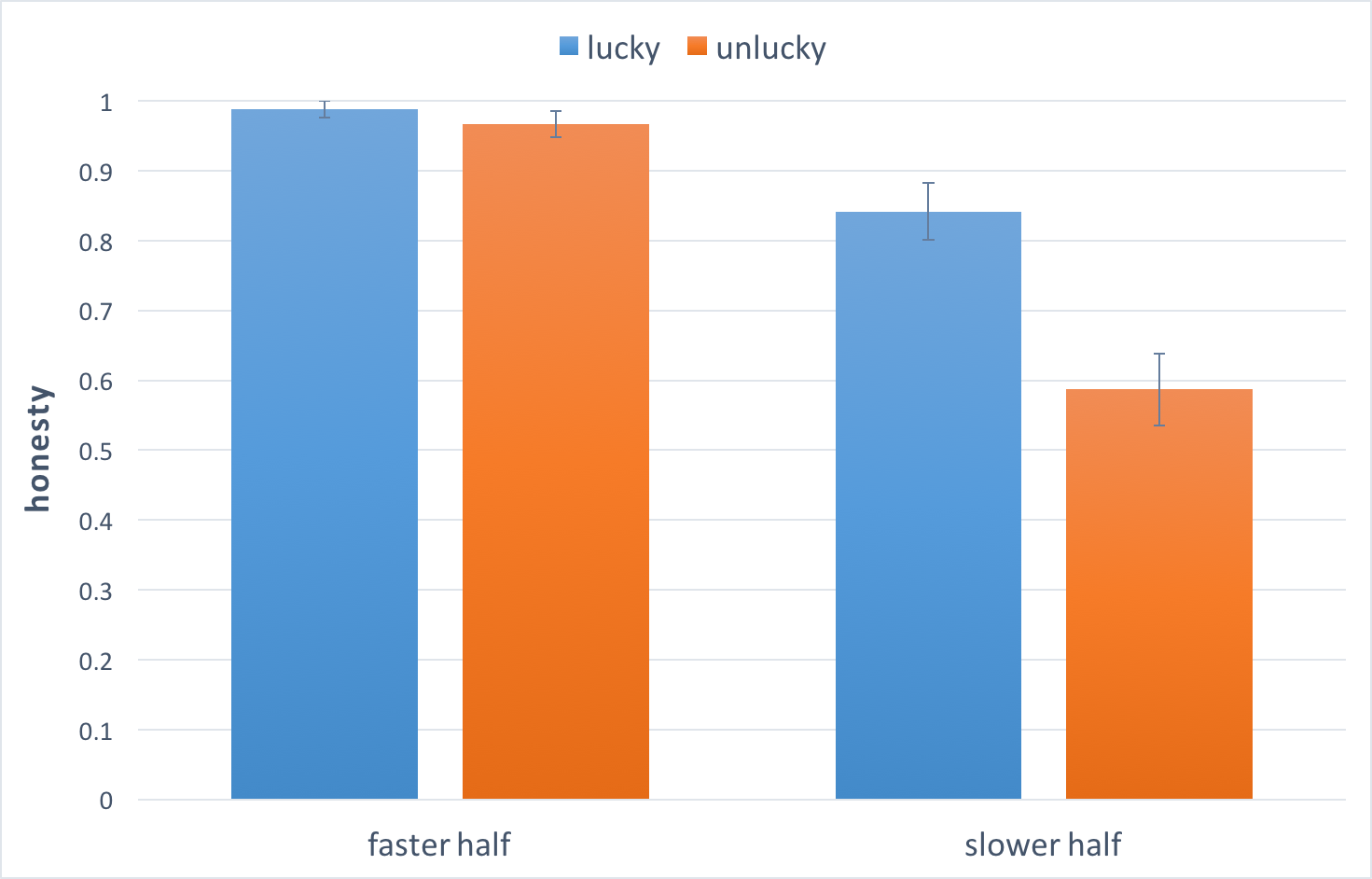}
\end{figure}

Fast participants in the Lucky condition behaved essentially the same as the Fast participants in the Unlucky condition, and they were extremely honest (rate of honesty: 98.8\% vs 96.7\%; logit regression without control: coeff $= 1.02$, z $= 0.87$, p $= 0.383$; with control: coeff $= 1.14$, z $= 0.96$, p $= 0.335$). On the contrary, among Slow participants, Unlucky participants were way more dishonest than Slow and Lucky ones (rate of honesty: 58.7\% vs 84.1\%; logit regression without control: coeff $= 1.32$, z $= 3.57$, p $< .001$; with control: coeff $= 1.30$, z $= 3.33$, p $= 0.001$). 

This suggests that, indeed, Fast participants were less likely than Slow participants to have found out the payoff corresponding to telling the truth and that being lucky has a positive effect on honesty, but only among participants who found out all the payoff corresponding to telling the truth. Moreover, the fact that virtually all Fast participants were honest suggests that there is a proportion of participants who tell the truth without first finding out its corresponding payoff.
\bigskip

\textbf{Result 2.} There is a proportion of participants who report the truth without first finding out its corresponding payoff.
\bigskip

\textbf{Result 3.} Among participants who decide to find out the payoff corresponding to telling the truth, this payoff has a positive effect on honesty.

\section{Limitation of Study 1}\label{se:limitation}

Study 1 clearly shows that participants vary in their moral cost of lying: there is a proportion of participants who report the truth without first finding out the corresponding payoff; and there is a proportion of participants who first find out the payoff corresponding to telling the truth and then lie only if this payoff is low.

However, in principle, the results of Study 1 can be explained by assuming that participants vary \emph{only} in the cost of lying, and that (virtually) all participants have zero cost of finding out payoffs. To see this, assume that all participants have zero cost of finding out the payoffs and let $\lambda(p)\geq0$ be the individual parameter representing the moral cost of lying of participant $p$. If $\lambda(p)$ is above a certain threshold $\lambda^*$, then $p$ knows that s/he will tell the truth regardless of the corresponding payoff (remember that participants know that their payoff is somewhere between 0 and 90 cents). Since the cost of finding out payoffs is assumed to be 0, such a participant $p$ will be indifferent between finding out the payoffs or not and therefore, statistically, $p$ will stop finding out the payoffs in the middle of the list. If, instead, $\lambda(p) < \lambda^*$, it is optimal for $p$ to find out all the payoffs. Therefore, participants $p$ with $\lambda(p) < \lambda^*$ will take longer to make a decision compared to participants with $\lambda(p) \geq \lambda^*$, because they have to find out all the payoffs. Moreover, participants with $\lambda(p) < \lambda^*$ will either tell the truth or lie maximally, depending on whether their cost of lying counterbalance or not the difference between the maximum available payoff and the payoff corresponding to telling the truth. These predictions are in line with the results of Study 1, which indeed show that virtually all participants either tell the truth or report the position corresponding to the global maximum (note that almost no one reported a position corresponding to a local maximum), and that Fast participants are way more honest than Slow participants. Therefore, the results of Study 1 can be explained without assuming that investing time to find out the payoffs has a cost for the participants.

The goal of Study 2 is to show that a significant proportion of participants do have a non-zero cost of finding out payoffs.

To do this, we observe that the main limitation of Study 1 is that there are no very high payoffs at the beginning of the list. In Study 2, we consider a new list characterized by the fact that, at the beginning, there is a very high payoff, which is however still smaller than the global maximum. The rationale for considering this variant is the following. If also with this new list most participants will either report the truth or report a position corresponding to the global maximum, then this would be a indication that virtually all people have a zero-cost of finding out payoffs. If, instead, in Study 2 a significant proportion of participants will report this high local maximum, this would be an indication that a significant proportion of subjects have a non-zero cost of finding out payoffs.

Furthermore, in Study 2 we would like to bring additional evidence that there is a proportion of participants who report the true position without first finding out its corresponding payoff. In Study 1 this result was obtained indirectly by analyzing response times. In Study 2 we would like to strengthen this result through an additional test. At the end of the experiment, we will ask participants to self-report whether they have read the list of the payoffs before making a decision, or not. Then we will compare the behavior of those who respond that they have read the list of payoffs with the behavior of those who respond that they have not. We are aware that non-incentivized self-reported questions like this one have their limitations (e.g., participants are not incentivized to answer and therefore they might respond at random). We indeed considered obtaining the same information through an incentivized question, for example by incentivizing participants to report, at the end of the survey, the payoff corresponding to their given position. However, we reasoned that such a measure would generate a bias, such that liars who have looked at the payoff corresponding to telling the truth would be hesitant to answer this question truthfully (even if it is incentivized), because it would be an admission that they have lied. Therefore, it is possible that some liars who have read the list would be classified as people who have not read the list. To avoid this bias we decided to adopt the self-reported question. In the results section we will show that, although it is a non-incentivized measure, it turns out to be reliable.

Finally, the third goal of Study 2 is to replicate the three main empirical regularities found in Study 1. Replication is crucial to increase confidence that findings are not false positives, and it is especially important now in light of the Replication Crisis in Psychology \citep{open2015estimating}.

\section{Study 2}\label{se:study2}

\subsection{Experimental design and procedure}

The design was identical to Study 1. There were only two differences. The first difference regarded the list of potential payoffs. In Study 2 we used the following list:

$$
65\,\,\,\, 87\,\,\,\, 27\,\,\,\, 36\,\,\,\, 62\,\,\,\, 20\,\,\,\, 53\,\,\,\, 54\,\,\,\, 64\,\,\,\, 59\,\,\,\, 47\,\,\,\, 75\,\,\,\, 60\,\,\,\, 30\,\,\,\, 57\,\,\,\, 53\,\,\,\, 75\,\,\,\, 27\,\,\,\, 16\,\,\,\, 40\,\,\,\, 86\,\,\,\, 71\,\,\,\, 11\,\,\,\, 88\,\,\,\, 5\,\,\,\, 42   
$$

We built this list as follows: apart from the payoffs corresponding to telling the truth in the Lucky and Unlucky situations, that we kept constant with respect to Study 1, this list was built from the list in Study 1, by taking the \emph{complementary} payoffs, that is, the payoffs in the list of Study 2 were equal to ninety minus the corresponding payoff in the list of Study 1. Note that this list fits our purpose: it contains a very high payoff at the beginning (87 in position 2), but this payoff is strictly smaller than the global maximum, which is 88, in position 24. Moreover, the payoffs corresponding to telling the truth are adjacent to larger payoffs. Indeed, in the Unlucky condition, reporting Position 20 instead of the true Position 19 guarantees a payoff of 40 cents instead of a payoff of 16 cents; similarly, in the Lucky condition, reporting Position 21 instead of Position 22 guarantees a payoff of 86 cents instead of a payoff of 71 cents. Therefore, similarly to Study 1, participants of Study 2 can ``stretch the truth'' very easily. Finally, the global maximum is not adjacent to either of the positions corresponding to telling the truth. 

The second difference with respect to Study 2 was that, after making their choice, participants were asked: ```Did you read the list before reporting the position?'' They could answer either yes or no.  

\subsection{Results}

\emph{Participants}

We recruited 200 participants on AMT, none of whom had participated in Study 1. After eliminating multiple IP addresses, multiple TurkIDs, those who left the survey incomplete, and those who failed the comprehension questions, we were left with $N=188$ participants (mean age = 34.3, females = 36.7\%).
\bigskip

\emph{Distribution of reported positions}

Figure 3 plots the histogram of the reported positions. As in Study 1, most subjects either report the truth or report a maximum of the payoff function, either local or global. Only a minority stretch the truth. More precisely, 30 subjects report position 2 (first local maximum), 1 subject reports position 5 (second local maximum), 1 subject reports position 9 (third local maximum), 131 subjects report the truth, 6 subjects report position 21 (but only 2 of them participated in the Lucky condition, and thus stretched the truth, according to our definition\footnote{Alternatively, one could define stretching the truth more flexibly, by including people who report a position one or two positions away from the true position. Clearly, this does not change the qualitative results.}), 1 subject reports position 23 (which is actually a minimum, so this person probably made a mistake when counting), 18 subject report position 24 (the global maximum). Therefore, Result 1 in Study 1 is replicated.

\begin{figure}[h]
  \caption{Histogram of the positions reported in Study 2.}
  \centering
    \includegraphics[width=0.7\textwidth]{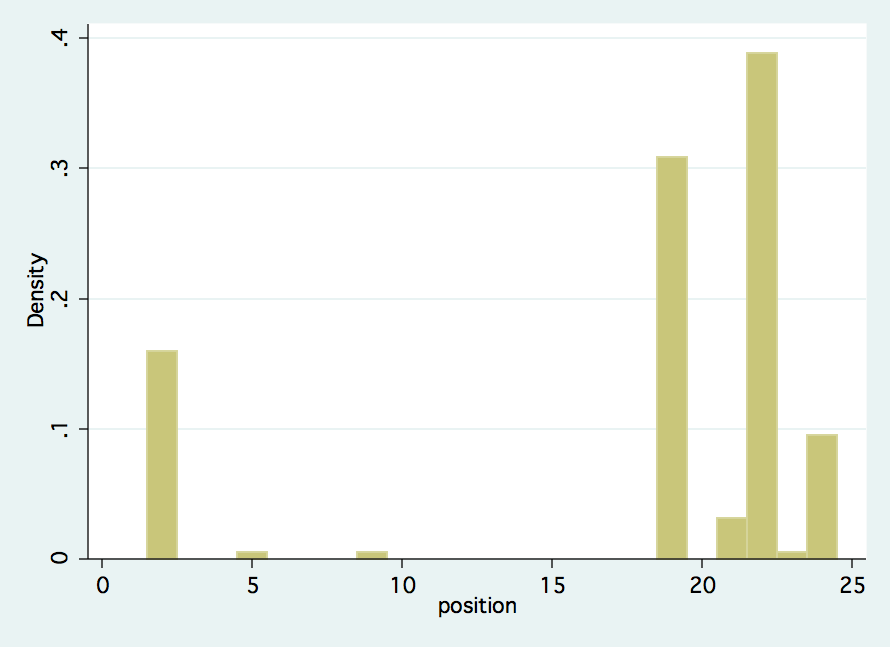}
\end{figure}

However, there is an important difference with respect to Study 1: while in Study 1 the third most chosen position was the global maximum, here the third most chosen position is position 2, which is a local but not global maximum. As mentioned in Section \ref{se:limitation}, this provides evidence that some people do have a non-zero cost of finding out payoffs. 

Another important difference with respect to Study 1 regards the rate of honesty. In Study 1 the rate of honesty was 84.1\%. In Study 2 the rate of honesty was 69.7\%. Logit regression confirms that the difference is significant (p $< .001$). We are aware that this logit regression is not formally correct, because there is no random assignment between Study 1 and Study 2. We nevertheless believe that it is an indication that lying is more pervasive in Study 2 than in Study 1. In retrospection, this is not surprising: it is reasonable to think that participants with an ``average'' (but not zero) cost of finding out payoffs and with an ``average'' (but not zero) cost of lying, would find out \emph{some} payoffs and then lie only if, among the payoffs that they have found out, there is a very high payoff. As we will see, this observation turns out to be a natural consequence of the model we present in Section \ref{se:model}. 

\bigskip

\emph{Heterogeneity in participants' decision-making}

As in Study 1, in this section we aim to understand whether all participants find out the payoff corresponding to telling the truth before reporting a position, or whether there is a proportion of participants that report a position without first finding out the payoff corresponding to telling the truth. As in Study 1, we start by splitting participants in the faster half and the slower half. Next we will analyze the responses to the self-reported question: ``Did you read the list before reporting a position?''. 

Figure 4 reports the rate of honesty among lucky and unlucky participants, split by whether they responded fast or slow. The figure puts in evidence that, similarly to Study 1, the positive effect of luck is present only among slow participants (p $= 0.020$) but not among fast participants (p $=0.562$). This suggests that, indeed, slow participants are more likely than fast participants to find out the payoff corresponding to telling the truth (otherwise, the positive effect of luck on honesty would have been present also among fast participants). There is, however, a major difference with respect to Study 1. In Study 1 fast participants were almost exclusively honest (rate of honesty = 97.7\%); in Study 2, the rate of honesty among fast participants appears to be quite smaller (80.8\%). To understand the reason of this difference, we look at the distribution of reported positions among fast participants in Study 2. We find that virtually all fast liars (17 out of 18) report the local maximum (Position 2); only one fast liar reports the global maximum. Therefore, while in Study 1 the Fast participants were virtually all honest; in Study 2, virtually all Fast participants are either honest or report the early local maximum. This suggests that, in Study 1, some fast participants actually read part of the list of the payoffs, but then reported the true position only because they did not find a high payoff at the beginning of the list. In sum, this suggests being in the faster half is not a proxy for not reading the list at all, but it is a proxy for reading \emph{part} of the list. 

\begin{figure}[h]
  \caption{Rate of honesty among participants in Study 2, split by whether they were Lucky or Unlucky and by response time (faster half vs slower half). Error bars represent standard errors of the means}
  \centering
    \includegraphics[width=0.7\textwidth]{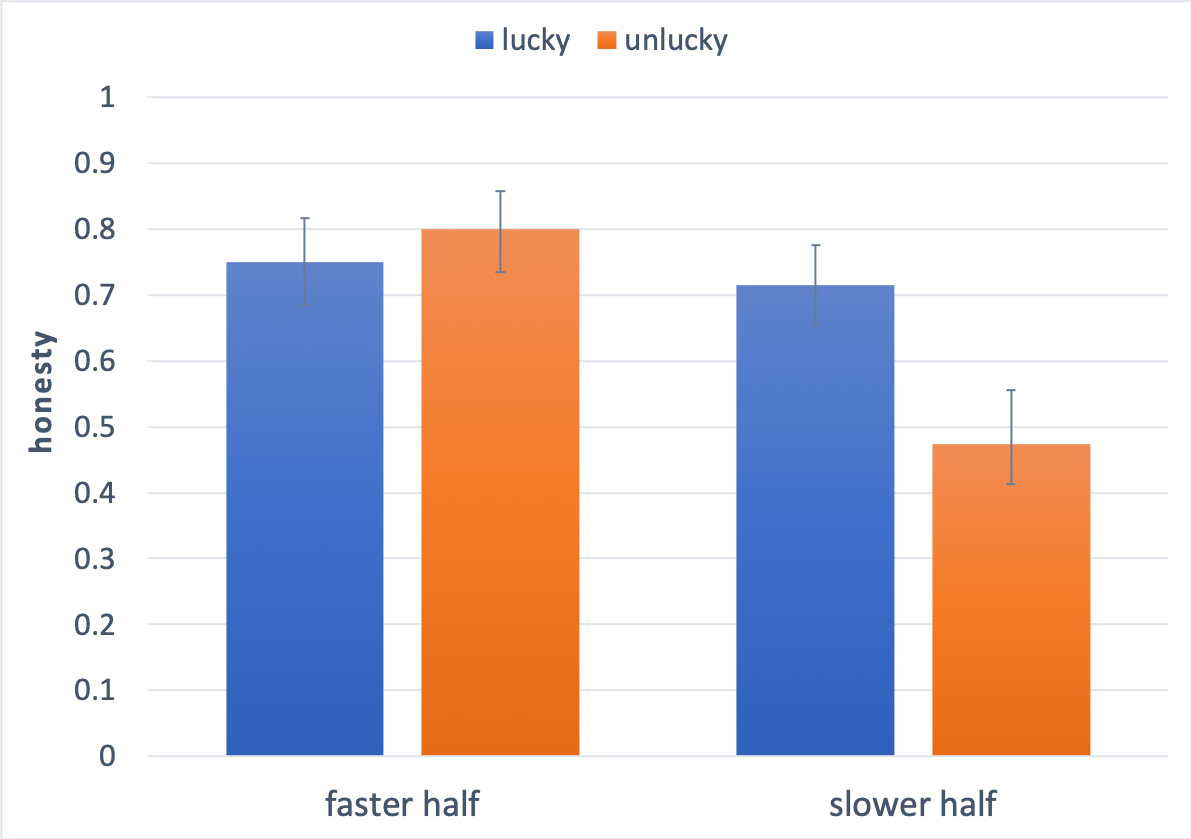}
\end{figure}

The analysis of answers to the self-reported question ``Did you read the list before reporting a position?'' confirms this view. We find that only 21.2\% of the participants declared that they did {\it not} read the list. Among these participants, the rate of honesty was 87.5\%, significantly higher than the rate of honesty among participants who declared that they had read the list, which was 64.8\% (p $=0.009$). Moreover, all the liars who declared that they had not read the list reported the local maximum (Position 2), except one who reported the next local maximum (Position 5). Furthermore, the average response time among participants who declared that they had not read the list (13.96s) was significantly smaller (p $<.001$) than the average response time among participants who declared that they had read the list (24.85s). Taken together, these results suggest that the answer to the question ``Did you read the list before reporting a position?'' was a reliable proxy for participants who read \emph{part} of the list.

We can now use these observations to show that Study 2 replicates also Result 2 and Result 3 of Study 1.

Result 2 stated that there is a significant proportion of participants who tell the truth without first finding out its corresponding payoff. To provide support for this result also in Study 2, we look at participants who declared that they had not read the list. The average response time among participants who reported the true position and then reported that they had not read the list was significantly smaller than the average response time among participants who reported the true position and then reported that they had read the list (12.46s vs 22.65s, p $<.001$). This provides a first piece of evidence that participants who reported the true position and declared that they had not read the list, reported the true position without first finding out its payoff corresponding. More critically, we find that participants who reported the true position and then reported that they had not read the list had a smaller response time than participants who did not report the true position and then reported that they had not read the list (12.46s vs 24.30s, p $=0.025$). Remembering that these liars almost exclusively reported Position = 2, the fact that participants who reported the true position take less time than these participants to report a position, suggests that most participants who reported the true position and then answered that they did not read the list, did not even start reading the list. This also implies that we can estimate that the proportion of subjects that report the true position without reading the list is, in our sample, 18.55\% (i.e., the proportion of honest people who declared that they had not read the list). 

Result 3 stated that, among participants who find out the payoff corresponding to telling the truth, being lucky has a positive effect on honesty. To provide support for this result also in Study 2, we look at the effect of luck among participants who declared that they had read the list of payoffs. As expected, among these participants, the rate of honesty among Lucky participants was significantly higher than the rate of honesty among Unlucky participants (61.4\% vs 38.5\%, p $=0.033$). Note that the same positive effect of Luck on honesty is not present among participants who declared that they had not read the list, in which case the results even trend in the opposite direction (rates of honesty: 77.8\% vs 95.4\%, p $=0.231$).

\section{Model}\label{se:model}

The goal of this section is to develop a model that can explain the empirical results found in the two experiments. We start by summarising these results in the next subsection. We then introduce an informal description of the model, which will help us guide through the formal development made in the following subsection. Next, we report the analysis of the model, that we divide in five propositions. Finally, we show that the model is consistent with the empirical results.

\subsection{Summary of the empirical results}\label{se:results}

The first three results correspond to the results found in Study 1 and replicated in Study 2.

\textbf{Result 1.} Very few participants stretch the truth, that is, the vast majority of participants either report the true position or a position that is \emph{not} adjacent to the true position; in this latter case, they almost always report a position corresponding to a maximum (either global or local).

\textbf{Result 2.} There is a significant proportion of participants who report the truth without first finding out its corresponding payoff.

\textbf{Result 3.} Among the participants who find out the payoff corresponding to telling the truth, luck has a positive effect on honesty (i.e., lucky participants lie less than unlucky participants).

The last result is the core finding of Study 2. 

\textbf{Result 4.} If there is a very high potential payoff (but smaller than the global maximum) at the beginning of the list, a substantial proportion of people report the position corresponding to this high payoff.

\subsection{Informal description of the model: A subdivision in nine types}\label{se:informal}

Before developing its formal details, we think it is useful to provide an informal description of the model. 

In particular, here we observe that the former four results can be broadly explained by assuming that there are nine types of participants, who differ on how they score along two dimensions: the intrinsic cost of lying, and the cost of finding out new payoffs. More precisely, assume that each participant can be classified by two individual variables $\lambda$ and $\tau$, where $\lambda$ is the intrinsic cost of lying and $\tau$ is the cost of finding out new payoffs. The cost of lying is paid only if a participant lies; the cost of finding out new payoffs is paid only if a participant finds out new payoffs. Apart from (possibly) paying these costs, we assume that participants aim at maximizing their material payoff. Let us assume that $\lambda$ and $\tau$ can assume only three values, that we denote $l$ (for low), $m$ (for medium), and $h$ (for high). For the sake of the informal argument, we assume that $l$ is essentially equal to 0, and we assume that $h$ is ``high enough'', meaning that a person with $h$ cost of lying will always tell the truth and a person with $h$ cost of finding out payoffs will report a position without finding out any of the payoffs. Therefore, we obtain a subdivision in 9 types as follows:

\begin{itemize}
    \item $(l,l)$ participants do not pay any cost for lying and pay no cost for finding out payoffs. In this case, the participant knows that s/he might lie, therefore, the best strategy is to find out all the payoffs and then report the position corresponding to the global maximum.
    \item $(l,m)$ participants do not pay any cost for lying and they pay a medium cost for finding out payoffs. Therefore, the best strategy for these participants is to read part of the list and then report the position corresponding to the maximum of the sub-list that they have read.
    \item $(l,h)$ participants do not pay any cost for lying, but they pay a high cost for finding out payoffs. Therefore, these participants do not find out any payoff and report a random position. 
    \item $(m,l)$ participants have medium cost of lying and pay no cost for finding out payoffs. In this case, the participant knows that s/he might lie, therefore, the best strategy is to find out all the payoffs and then lie only if the difference between the global maximal payoff and the payoff corresponding to telling the truth is greater than the cost of lying. 
    \item $(m,m)$ participants pay a medium cost for lying and a medium cost for finding out payoffs. Therefore, these participants will read part of the list and then, if they find a very high payoff in the part of the list that they have read, they lie, otherwise they tell the truth (observe that participants may or may not know the payoff corresponding to telling the truth; if they know it, clearly, the larger this payoff, the less likely will they be to lie). The best strategy will depend on M and will be computed exactly in section \ref{suse:analysis}.
    \item $(m,h)$ participants pay a medium cost for lying and a high cost for finding out payoffs. Therefore, the best strategy for these participants is to not find out any payoff and tell the truth.
    \item $(h,l)$ participants have high cost of lying and pay no cost for finding out payoffs. These participants know that they will tell the truth regardless of its payoff. Moreover, since the cost of finding out payoffs is zero, these participants find out the payoffs up to a random position, and then report the truth regardless of the payoffs that they have found out.
    \item $(h,m)$ participants pay a high cost for lying and a medium cost for finding out payoffs. These participants would, in principle, read part of the list and then report the true position, regardless of the payoffs that they have read. Since the cost of finding out new payoffs is not zero and since these participants already know that they will tell truth, then their best strategy is to report the truth without finding out any payoff.
    \item $(h,h)$ participants pay a high cost for lying and a high cost for finding out payoffs. The best strategy for these participants is to not find out any payoff and report the truth.
\end{itemize}

We note that this classification easily explains our results. Result 1 holds essentially by definition, because the cost of lying does not depend on the magnitude of the lie. Therefore, whenever a participant lies, he does so to maximize the payoff (globally, if he found out all the payoffs; locally, if he found out only some payoffs). Result 2 (a substantial proportion of people report the truth without first finding out the corresponding payoff) comes from participants in the $(m,h)$, $(h,m)$, and $(h,h)$ classes; Result 3 (among those finding out the payoff corresponding to telling the truth, luck has a positive effect) comes from participants in the $(m,l)$ class; and Result 4 (if there is a high payoff at the beginning of the list, a substantial proportion of people report the corresponding position) comes from the $(m,m)$ class. Of course, we are not saying that all these nine classes are non-empty, or that they have the same density (for example, since we found no evidence of random play, the class $(l,h)$ is likely to be empty, or negligible). The only point that we want to make is that this is a relatively simple way to explain our results, and that we can use it as guiding idea to develop the formal model.

\subsection{The model}

Here we introduce the formal model. We use a sequential problem in which, at each stage, a participant has to decide between reporting a position, or finding out the payoff corresponding to the next position. Finding out a payoff has a cost (that might be $0$), so as it has a cost to lie (that might be $0$, as well).

More precisely, we consider a sequential decision problem as follows. At time $t_0=0$, participant $p$ is given the following pieces of information:
\begin{itemize}
\item A private state of the world $\pi_g\in\Pi=\{\pi_1,\pi_2,\ldots,\pi_n\}$, where $\Pi$ is a finite set representing the potential states of the world. In our experiments, $\Pi$ is the set of all positions and $\pi_g$ is the actual position given to $p$. At time $t=0$, $p$ need not know the exact composition of the set $\Pi$. For example, in our experiments participants do not initially know the length of the list of potential payoffs.
\item A set $m(\Pi)$ of cardinality $\leq n$ of potential material payoffs (expressed in some unit of measurement, e.g., cents). At time $t=0$, $p$ knows that each payoff is drawn at random from the set $\{1,2,\ldots, M\}$, where $M$ is a positive integer (in our experiments $M=90$ cents).
\end{itemize}
At time $t_0=0$, $p$ can decide between two options: reporting a state of the world without knowing any payoff, or paying a cost to find out the payoff corresponding to the first state of the world. More precisely, the two available options at time $t_0=0$ are:

\begin{itemize}
\item \emph{Option R$_0$. Report a state of the world $\pi_r$ at time $t=0$}. The participant $p$ can either report the state of the world truthfully ($\pi_r=\pi_g$) or not ($\pi_r\neq\pi_g$). We assume that participants have a moral cost of lying (that might be $0$). We include the moral cost of lying into a utility function taking inspiration by Krupka \& Weber (2013). More precisely, we define the expected utility of reporting $\pi_r$ at time $t_0=0$ to be:
$$
U_p(\pi_r,\pi_g,0) = \frac{M+1}{2} - \lambda(p)(1-\chi_{\pi_r,\pi_g}),
$$
where 
$$
\chi_{\pi_r,\pi_g}=
\begin{cases}
1, \qquad\text{if } \pi_r=\pi_g\\
0, \qquad\text{if } \pi_r\neq\pi_g
\end{cases}
$$
The individual parameter $\lambda(p)\geq0$ represents the intrinsic cost of lying, which we assume to be an integer. We assume that (for virtually all participants) $\lambda$ does not depend on the magnitude of the lie, because Result 1 suggests that very few people stretch the truth. The interpretation of $\lambda(p)$ is clear: $p$ lies if the benefit of lying is greater than  $\lambda(p)$; if the benefit of lying is equal to $\lambda(p)$, then $p$ is indifferent. So, the expected utility of reporting $\pi_r$ at time $t_0=0$ is equal to the expected payoff of choosing to report a state of the world knowing that its corresponding payoff is drawn at random from the set $\{1,2,\ldots, M\}$, minus the intrinsic cost of lying, which is paid only when $p$ indeed lies (i.e., $\pi_r\neq\pi_g$).  Thus, the term $\frac{M+1}{2}$ comes from the fact that, at time $t_0=0$, $p$ does not know any of the payoffs and therefore $p$'s expected monetary payoff from reporting $\pi_r$ is $\frac{M+1}{2}$, independently of $\pi_r\in\Pi$.
\item \emph{Option F$_0$. Choose to pay a cost $\tau_1(p)$ to find out the first payoff $m(\pi_1)$ at time $t_1$}. We assume that the cost $\tau_1(p)\geq0$ is a integer parameter, depending only on the participant. 
\end{itemize}

And so on, at time $t=1,2,\ldots$ (we assume that time is discrete), if $p$ has not chosen to report a state of the world $\pi_r$ at time $t-1$, then $p$ will have to choose between two options:

\begin{itemize}
\item \emph{Option R$_t$. Report a state of the world $\pi_r$ at time $t$}. The expected utility of reporting $\pi_r$ at time $t$ is defined as:
$$
U_p(\pi_r,\pi_g,t) = \mathbb E(\pi_r,\pi_t)  - \lambda(p)(1-\chi_{\pi_r,\pi_g}) - \tau_1(p) - \ldots - \tau_{t}(p),
$$
where the expected material payoff $\mathbb E(\pi_r,\pi_t)$ of reporting $\pi_r$ after finding out $t$ payoffs is:
$$
\mathbb E(\pi_r,\pi_t)=
\begin{cases}
m(\pi_r), \qquad\text{if } t\geq r\\
\frac{M+1}{2}, \qquad\text{if } t< r
\end{cases}
$$
The interpretation of $\mathbb E(\pi_r,\pi_t)$ is obvious. If $t<r$, then $p$ does not know, at time $t$, the material payoff corresponding to reporting $\pi_r$. Therefore, reporting $\pi_r$ at time $t<r$ has expected material payoff $\frac{M+1}{2}$. Else if $t\geq r$, then $p$ knows, at time $t$, the material payoff corresponding to reporting $\pi_r$. Therefore, in this case, the expected material payoff of reporting $\pi_r$ at time $t$ is equal to the actual payoff associate to $\pi_r$, which is $m(\pi_r)$.
\item \emph{Option F$_t$. Choose to pay a cost $\tau_{t+1}(p)$ to find out the payoff $m(\pi_{t+1})$ at time $t+1$}. We assume that $\tau_{t+1}(p)\geq\tau_t(p)$. This property will be used in the proofs. We believe it to be a reasonable assumption, coming from the observation that memorizing $t+1$ payoffs may require more cognitive effort than memorizing $t$ payoffs. 
\end{itemize}

Our goal is to analyze the strategies that maximize people's expected utility. To this end, we will make the following two assumptions:

\begin{enumerate}
    \item Players are risk neutral
    \item Players are rational
\end{enumerate}

The first assumption is technically useful, but it does not change the core of the argument: the model can be easily adapted to risk averse (or risk seeking) players, by adding an additional parameter, which, everything else being equal, will make these players less (or more) likely to find out a new payoff. The second assumption is the standard assumption of rationality in game theory. In our case, it implies that, at time $t-1$, $p$ decides to find out $m(\pi_t)$ only if the expected utility of finding out $m(\pi_t)$ {\it and} deciding to report a state of the world at time $t$ is greater than the utility of reporting a state of the world at time $t-1$. Indeed, at time $t-1$, $p$ always prefer to find out $m(\pi_t)$ {\it and} choose to report a position at time $t$ over finding out $m(\pi_t)$ and then choosing to {\it not} report a position at time $t$ (because, in the latter case, $p$ will have to pay the additional cost $\tau_{t+1}$ to find out the next payoff). 

\subsection{Analysis}\label{suse:analysis}

We start by discussing the model in two extreme cases: $\tau_1(p)=\tau_2(p)=\ldots=0$ and $\lambda(p)=0$. Then we move to the general analysis. Here we report only the propositions along with a brief description. We refer to the Appendix for the proofs. 

\subsubsection{Case $\tau_1(p)=\tau_2(p)=\ldots=0$}

We make one technical assumption in order to avoid listing a number of trivial subcases, that differ from one another in only minor details. Specifically, let $\overline M = \max_{i=1,\ldots,n}m(\pi_i)$, we assume that $\overline M < M$.  This assumption guarantees that $p$ has always a non-zero probability of finding out a payoff larger than those that he has already found out, and therefore $p$ would not stop finding out payoffs simply because he has found the maximum available payoff. Note that this assumtpion is satisfied by our experiments, since $M=90$, but none of the numbers in the list of potential payoffs is equal to 90. 

\begin{proposition}\label{prop:time_cost_zero}
{\rm Let $p$ be a participant with $\tau_1(p)=\ldots=\tau_n(p)=0$. Then:
\begin{itemize}
\item If $\lambda(p) > M$, that is, if the intrinsic cost of lying is greater than the maximal payoff, $p$ maximizes their utility by finding out a random number of payoffs and by reporting the true state of the world, regardless of the payoffs that $p$ has found out;
\item If $\lambda(p)\leq M$, $p$ maximizes their utility by finding out all payoffs (note that this includes $\overline{M}$), and then:
\begin{itemize}
    \item reporting the truth, if $m(\pi_g) > \overline M - \lambda(p)$;
    \item reporting a state of the world $\pi$ such that $m(\pi)=\overline M$, if $m(\pi_g) < \overline M - \lambda(p)$;
    \item reporting a state of the world chosen at random between $\pi_g$ and $\pi$ such that $m(\pi)=\overline M$, if $m(\pi_g) = \overline M - \lambda(p)$.
\end{itemize}
\end{itemize}
}
\end{proposition}

The intuition behind this proposition is very simple. The case $\lambda(p)> M$ is trivial: $p$ knows that no material payoff would counterbalance the intrinsic cost of lying, therefore $p$ will always tell the truth, after finding out a random number of payoffs (because $p$ does not pay any cost for finding out payoffs). Else, if $\lambda(p)\leq M$, then $p$ might lie depending on the payoffs. The fact that $p$ does not pay any cost for finding out payoffs and the fact that there is always a non-zero probability of finding out a larger payoff, imply that $p$'s best strategy is to find out all payoffs before reporting a state of the world. At this point, $p$ will report the truth or the state of the world corresponding to the maximum payoff, depending on the cost of lying: if the cost of lying is above a certain threshold, then $p$ will tell the truth; if the cost of lying is below that threshold, $p$ will lie; if the cost of lying is equal to the threshold, $p$ will be indifferent between telling the truth or lying. 

\subsubsection{Case $\lambda(p)=0$}

The case $\lambda(p)=0$ (that is, $p$ has no cost of lying) is intuitively simple. Since $p$ has no cost of lying, then $p$ will find out some payoffs (depending on $\tau_t(p)$) and then will report a position that maximizes their (possibly expected) material payoff. 

To formalize this idea, we need to compute the time at which $p$ stops finding out new payoffs. To do so, let $M_t = \max_{1,2,\ldots,t}\{m(\pi_1),\ldots,m(\pi_t),\lfloor\frac{M+1}{2}\rfloor\}$, where $\lfloor x\rfloor$ denote the highest integer which is smaller or equal than $x$. Observe that the sum $\frac{1}{M}\left(\sum_{k= M_t+1}^Mk\right)$ is weakly decreasing as $t$ increases, while the sequence $\tau_t(p)$, by assumption, is weakly increasing in $t$.  It follows that, if there is a $\overline t$ such that 

$$\frac{1}{M}\left(\sum_{k= M_t+1}^Mk\right)\geq\tau_t,\text{ for all } t\leq\overline t\qquad\text{and}\qquad\frac{1}{M}\left(\sum_{k= M_t+1}^Mk\right)<\tau_t,\text{ for all } t>\overline t
$$
then $\overline t$ is unique. We introduce the following notation: if $\overline t$ as just defined exists, we set $t^{last}=\overline t$; if, instead, $\overline t$ does not exist, we set $t^{last}=n$.

\begin{proposition}\label{prop:lying_cost_zero}
{\rm Suppose that $\lambda(p)=0$. The strategy that maximizes $p$'s utility is to first find out all payoffs up to $t^{last}$ and then report a state of the world according to the following three cases:
\begin{enumerate}
    \item if $t^{last}<n$ and $m(\pi_i)<\frac{M+1}{2}$, for all $i\leq t^{last}$, then $p$ maximizes their utility by reporting a state of the world, $\pi_r$, chosen at random among those for which $r>t^{last}$; 
    \item if $t^{last}<n$ and there exists a $\pi_i$, with $i\leq t^{last}$, such that $m(\pi_i)\geq\frac{M+1}{2}$, then  player $p$ maximizes their utility by reporting a state of the world $\pi_r$ such that $m(\pi_r)=\max_{i=1,\ldots,t^{last}}m(\pi_i) = M_{t^{last}}$.
    \item if $t^{last}=n$ then  player $p$ maximizes their utility by reporting a state of the world, $\pi_r$, such that $m(\pi_r)=\max_{i=1,\ldots,n}m(\pi_i)=\overline M$.
\end{enumerate}

}
\end{proposition}

\subsubsection{General case: $\lambda(p)>0$ and $\tau_t(p)>0$, for all $t$}

Now we analyze the model for $\lambda(p)>0$ and $\tau_t(p)>0$, for all $t$. (Clearly, the case $\tau_t(p)=0$ up to some $\overline t < n$, and $\tau_t(p)>0$, for $t>\overline{t}$ will be a combination of the previous case and this general case). Moreover, since $p$ will always be a fixed player, to simplify the notation, we will write $\lambda$ instead of $\lambda(p)$ and $\tau_t$ instead of $\tau_t(p)$. 

We start by looking at what happen at time $t_0=0$. We keep this case separated from the general case for two reasons. One is technical: at time $t_0=0$, $p$ does not know any of the payoffs, and therefore the details are slightly different than in the case $t>0$, in which $p$ knows at least one payoff. The other one is conceptual, because we want to explicitly show that there are values of $\lambda$ and $\tau$ such that people report the truth at time $t_0$, and this would be consistent with the analysis in Section \ref{se:results}, which suggests that not only there is a proportion of participants who report the truth without first finding out its corresponding payoff, but there is even a proportion of people who report the truth without finding out {\it any} payoff. 

\begin{proposition}\label{prop:time_zero}
{\rm It is optimal for participant $p$ to report the true state of the world $\pi_g$ at time $t_0=0$ if one of the following conditions is satisfied:
\begin{itemize}
    \item $\lambda\geq\frac{M+1}{2}$ and $\tau_1>0$;
    \item $g=1$, $\lambda<\frac{M+1}{2}$ and $\tau_1>\frac{1}{M}\sum_{k=\lfloor\frac{M+1}{2}\rfloor-\lambda+1}^M \left(k-\frac{M+1}{2}+\lambda\right)$;
    \item $g>1$, $\lambda<\frac{M+1}{2}$ and $\tau_{1}>\frac{1}{M}\sum_{k=\lfloor\frac{M+1}{2}\rfloor+\lambda+1}^M\left(k-\lambda-\frac{M+1}{2}\right)$.
\end{itemize}
}
\end{proposition}

We remind that, in our two experiments, $g>1$. It follows that, according to this proposition, there are two classes of participants who report the true position without first finding out any payoff: participants with high cost of lying and non-zero cost of finding out payoffs, and participants with a medium cost of lying, but a high enough cost of finding out payoffs. This would correspond, in the informal description of the model in Section \ref{se:informal}, to $(h,m)$, $(h,h)$ and $(m,h)$ participants.

We can now present the general case. Notice that, if the starting point of a sum is greater than its ending point, we follow the standard convention and we define the sum to be equal to zero.

  \begin{proposition}\label{prop:general}
  {\rm Suppose that $p$ has found out all payoffs up to $m(\pi_t)$, then:
  \begin{itemize}
      \item $p$ maximizes his expected utility by reporting, at time $t$, the true state of the world, if one of the following conditions is satisfied:
      \begin{itemize}
          \item $g\leq t$, $m(\pi_g)> M_t-\lambda$, and $
     \tau_{t+1}>\frac{1}{M}\sum_{k=m(\pi_g)+\lambda+1}^M(k-\lambda-m(\pi_g))
     $; or 
     \item $g=t+1$, $M_t-\lambda<\frac{M+1}{2}$, and $
         \tau_{t+1}>\frac{1}{M}\sum_{k=1}^{M_t-\lambda}\left(M_t-\lambda-\frac{M+1}{2}\right)+\frac{1}{M}\sum_{k=M_t-\lambda+1}^{M}\left(k-\frac{M+1}{2}\right)
         $; or
        \item $g>t+1$, $M_t-\lambda<\frac{M+1}{2}$, and  $
         \tau_{t+1}>\frac{1}{M}\sum_{k=\lfloor\frac{M+1}{2}\rfloor+\lambda+1}^M\left(k-\lambda-\frac{M+1}{2}\right)
         $.
      \end{itemize}

      \item $p$ maximizes his expected utility by reporting, at time $t$, a state of the world corresponding to a payoff of $M_t$ (if there is more than one, then $p$ chooses at random among them), if one of the following conditions is satisfied:
      \begin{itemize}
          \item $g\leq t$, $m(\pi_g)< M_t-\lambda$, and $
     \tau_{t+1}>\frac{1}{M}\sum_{k=m(\pi_g)+\lambda+1}^M(k-\lambda-m(\pi_g))
     $; or 
     \item $g=t+1$, $M_t-\lambda>\frac{M+1}{2}$, and $
         \tau_{t+1}>\frac{1}{M}\sum_{k=1}^{M_t-\lambda}\left(M_t-\lambda-\frac{M+1}{2}\right)+\frac{1}{M}\sum_{k=M_t-\lambda+1}^{M}\left(k-\frac{M+1}{2}\right)
         $; or
        \item $g>t+1$, $M_t-\lambda>\frac{M+1}{2}$, and  $
         \tau_{t+1}>\frac{1}{M}\sum_{k=\lfloor\frac{M+1}{2}\rfloor+\lambda+1}^M\left(k-\lambda-\frac{M+1}{2}\right)
         $.
      \end{itemize}
      \item In all other situations, $p$ maximizes his expected utility by paying, at time $t$, a cost $\tau_{t+1}$ to find out $m(\pi_{t+1})$ at time $t+1$.
  \end{itemize}
  }
  \end{proposition}
  
  \subsubsection{The predictions of the model are consistent with the empirical results}
  
  We conclude by observing that the previous propositions are consistent with the qualitative empirical results listed in Section \ref{se:results}. 
  
  \bigskip
  
  \emph{Consistency with Result 1}
  
  All propositions above affirm that the state of the world that maximizes participants' expected utility is either the true state of the world, or a state of the world that maximizes the payoff of a sublist of the shape $\{m(\pi_1),\ldots, m(\pi_t)\}$, with $t\leq n$.
 
This is broadly consistent with Result 1. In particular, the fact that the model predicts that no one stretches the truth ultimately follows from the assumption that the cost of lying $\lambda(p)$ does not depend on some distance between the lie and the truth. Had we assumed that $\lambda(p)$ depends on the magnitude of the lie, the predictions would have been different. Our empirical data show that a small proportion (less than 5\%) of subjects stretch the truth. These subjects can be covered by the model by assuming that, for these subjects, the cost of lying depends on the distance between the lie and the truth.

\bigskip

\emph{Consistency with Result 2}
  
Proposition \ref{prop:time_zero} and Proposition \ref{prop:general} show that there are values of $\lambda(p)$ and $\tau(p)$ for which it is in $p$'s best interest to report the truth without finding out the corresponding payoff. This is consistent with Result 2. 

We note, additionally, that Proposition \ref{prop:time_zero} shows that there are values of $\lambda(p)$ and $\tau(p)$ for which it is in $p$'s best interest to report the truth without finding out {\it any} payoff. This is consistent with the analysis in Section \ref{suse:analysis} regarding Study 2, which shows that the average response time of people who report the true position and declare that they have not read the list is smaller than the average response time of people who report Position 2 (remember that, in Study 2, Position 2 corresponds to a high local maximum), suggesting that the people who report the true position and declare that they have not read the list, did not even start reading the list.

\bigskip

\emph{Consistency with Result 3}

The third result states that, among subjects who find out the payoff corresponding to telling the truth, $m(\pi_g)$, this payoff has a positive effect on honesty.

To see whether this result is consistent with the model, we have to look at two separate cases, depending on whether the costs of finding out payoffs are all zero or not. 

If $\tau_1(p)=\ldots=\tau_n(p)=0$, Proposition \ref{prop:time_cost_zero} guarantees that the greater $m(\pi_g)$ is, the more likely is $p$ to report the truth, which is indeed consistent with Result 3. 

Else if $\tau_i(p)>0$, for some $i$ that we can assume to be, without loss of generality, equal to $1$, we have to look at Proposition \ref{prop:general}, case $g\leq t$ (i.e., $p$ has found out all payoffs up to at least $m(\pi_g)$). In this case, we note that the condition for reporting the truth depends on $m(\pi_t)$ in such a way that, when $m(\pi_t)$ increases, the condition is more likely to be satisfied. Thus, also this proposition predicts a positive effect of $m(\pi_g)$ on honesty among subjects who found $m(\pi_g)$ out. 

\bigskip

\emph{Consistency with Result 4}

Result 4 states that, if there is a high potential payoff (but smaller than the global maximum) at the beginning of the list, then a significant proportion of subjects report the position corresponding to this high payoff. 

To see whether this result is consistent with the model, we have to distinguish two cases, depending on whether the cost of lying of a participant is equal to zero or not. 

If $\lambda(p)=0$, Proposition \ref{prop:lying_cost_zero} states that $p$ reports a state of the world corresponding to the maximum of the sublist $\{m(\pi_1),\ldots, m(\pi_{t^{last}})\}$. Now, remember that what we denoted as $t^{last}$ depends negatively on $M_t=\max\{m(\pi_1),\ldots,m(\pi_t),\lfloor\frac{M+1}{2}\rfloor\}$. Therefore, putting a high potential payoff at the beginning of the list, makes $M_t$ increase, which in turn makes $t^{last}$ decrease and, by Proposition \ref{prop:lying_cost_zero}, makes $p$ more likely to report the position corresponding to $M_t$.  

Else if $\lambda(p)>0$, Proposition \ref{prop:general} states that $p$ reports a state of the world corresponding to $M_t$ if certain conditions are satisfied. We need to show that these conditions depend positively on $M_t$, that is, larger $M_t$'s make these conditions more likely to be satisfied. Three of these conditions ($m(\pi_g)< M_t-\lambda$, $M_t-\lambda>\frac{M+1}{2}$, and $M_t-\lambda>\frac{M+1}{2}$), obviously depend positively on $M_t$. The only non-trivial condition is 
$$
\tau_{t+1}>\frac{1}{M}\sum_{k=1}^{M_t-\lambda}\left(M_t-\lambda-\frac{M+1}{2}\right)+\frac{1}{M}\sum_{k=M_t-\lambda+1}^{M}\left(k-\frac{M+1}{2}\right).
$$

To show that this condition depends positively on $M_t$, we have to show that its right hand side decreases with $M_t$. Note that, when passing from $M_t$ to $M_t+1$,
 the first summand of the right hand side increases by $(M_t-\lambda-\frac{M+1}{2})/M$, while its second summand decreases by $(M_t-\lambda+1-\frac{M+1}{2})/M.$ Since the decrease of the right hand side is greater than its increase, then, overall, the right hand side decreases as $M_t$ goes to $M_t+1$. This implies that the right hand side decreases with $M_t$, as wanted.

  \section{Discussion}\label{se:discussion}
  
  Previous work studied (dis)honesty in situations in which the decision maker knows the economic consequences of lying vs telling the truth (e.g., die-under-cup paradigm) or in situations in which any uncertainty about these consequences cannot be cleared in anyway (e.g., Sender-Receiver game). However, in reality, people often know that they \emph{will} have a chance to lie, they do not initially know the economic consequences of the available choices, but they can invest time to find them out. 
  
  The first goal of this work is to capture the core of this kind of situations by means of a novel decision problem. Participants are initially given a list of potential payoffs and a position on this list. The position is their private information. Participants are told that their job is to report their position and that they will be paid an amount of money equal to the payoff in the position they report. Therefore, participants can lie by reporting a position corresponding to a payoff higher than the payoff in their given position.
  
  We conducted two experiments, which provided evidence of four major results: (i) Very few participants stretch the truth, that is, the vast majority of participants either report the true position or a position that is \emph{not} adjacent to the true position - in this latter case, they almost always report a position corresponding to a maximum (either global or local); (ii) There is a significant proportion of participants who report the true position without first finding out its corresponding payoff; (iii) Among the participants who find out the payoff corresponding to reporting the true position, this payoff has a positive effect on honesty; (iv) If there is a high potential payoff (but smaller than the global maximum) at the beginning of the list, a significant proportion of participants report the position corresponding to that high payoff.
  
  These results suggest that participants vary in (at least) two dimensions: the intrinsic cost of lying, and the cost of finding out new payoffs. 
  
  Starting from this observation, we built a model in which participants have to decide at what time $t^{last}$ to stop finding out payoffs, and what position $\pi_r$ to report. The utility is assumed to depend on the (expected) material payoff corresponding to reporting a position at a given time, on whether they lie or not (if so, they pay the intrinsic cost of lying), and on the cost of finding out the payoffs. We have shown that this model is consistent with the four empirical results. 

 Therefore, our work makes a step forward with respect to previous research from both the empirical and the theoretical viewpoints. Empirically, we reported the first studies exploring dishonesty in situations in which people know that they will have a chance to lie, they do not initially know the economic consequences of the available actions, but they can invest time to find them out. Theoretically, we proposed a model that is consistent with observed behavior in these situations.   
 
Additionally, our results contribute also to the debate on whether the intrinsic cost of lying depend on the magnitude of the lie or not. Mazar, Amir and Ariely (\citeyear{mazar2008dishonesty}) and Fischbacher and F\"olmi-Heusi (\citeyear{fischbacher2013lies}) advanced the hypothesis that the marginal cost of the lie is increasing in the magnitude of the lie, which implies that people might lie a bit, but not maximally. \cite{gneezy2018lying} recently found that most people lie maximally, but there are also people who lie partially, and this partly depends on whether reputation is at stake. Our results are broadly in line with this view. In our experiments, most liars lied maximally, but there was also small proportion of participants stretching the truth (below 5\%).

Finally, our results also contribute to the debate on whether honesty is intuitive or deliberative. Indeed, we have also conducted a time pressure study (analyzed in the Appendix). This is a significant contribution in itself, because the role of time pressure and, more generally, the role of intuition on honesty has been at the center of the debate in the last years \citep{andersen2018allowing,cappelen2013we,capraro2017does,capraro2019time,debey2012lying,gino2011unable,gunia2012contemplation,lohse2018deception,mead2009too,shalvi2012honesty,spence2001behavioural,van2014limited,walczyk2003cognitive}, as a part of the more general research program of classifying social behaviors according to whether they are intuitive or reflective \citep{capraro2015social,capraro2016rethinking,capraro2017deliberation,corgnet2015cognitive,lotz2015spontaneous,rand2012spontaneous,rand2014social,rand2016social,rand2016cooperation}. In our case, time pressure has no effect on the overall rate of honesty. However, it does have an effect on the distribution of choices: (i) it transforms some would-be global maximizers into local-maximizers, and (ii) it generates a small proportion of confused subjects. 

In sum, we introduced a novel decision problem to study dishonesty in situations in which people know that they will have a chance to lie, they do not initially know the economic consequences of lying vs telling the truth, but they can invest time to find them out. We also introduced a model, according to which people pay a cost of lying (if they lie) and a cost of finding out the economic consequences (if they find them out), that is consistent with the empirical results.

  \section{Acknowledgement}

This material is based upon work supported by the National Science Foundation under Grant No. 1440140, while both authors were in residence at the Mathematical Sciences Research Institute in Berkeley, California.

\pagebreak

\bibliographystyle{apa}
\bibliography{lying}

\begin{thebibliography}{}

\bibitem[\protect\astroncite{Abeler et~al.}{2019}]{abeler2018preferences}
Abeler, J., Nosenzo, D., and Raymond, C. (2019).
\newblock Preferences for truth-telling.
\newblock {\em Econometrica}, 87:1115--1153.

\bibitem[\protect\astroncite{Andersen et~al.}{2018}]{andersen2018allowing}
Andersen, S., Gneezy, U., Kajackaite, A., and Marx, J. (2018).
\newblock Allowing for reflection time does not change behavior in dictator and
  cheating games.
\newblock {\em Journal of Economic Behavior \& Organization}, 145:24--33.

\bibitem[\protect\astroncite{Arechar et~al.}{2018}]{arechar2018conducting}
Arechar, A.~A., G{\"a}chter, S., and Molleman, L. (2018).
\newblock Conducting interactive experiments online.
\newblock {\em Experimental Economics}, 21(1):99--131.

\bibitem[\protect\astroncite{Berinsky et~al.}{2012}]{berinsky2012evaluating}
Berinsky, A.~J., Huber, G.~A., and Lenz, G.~S. (2012).
\newblock Evaluating online labor markets for experimental research:
  {A}mazon.com's {M}echanical {T}urk.
\newblock {\em Political Analysis}, 20(3):351--368.

\bibitem[\protect\astroncite{Biziou-van Pol et~al.}{2015}]{biziou2015does}
Biziou-van Pol, L., Haenen, J., Novaro, A., Occhipinti~Liberman, A., and
  Capraro, V. (2015).
\newblock Does telling white lies signal pro-social preferences?
\newblock {\em Judgment and Decision Making}, 10:538--548.

\bibitem[\protect\astroncite{Cappelen et~al.}{2013}]{cappelen2013we}
Cappelen, A.~W., S{\o}rensen, E.~{\O}., and Tungodden, B. (2013).
\newblock When do we lie?
\newblock {\em Journal of Economic Behavior \& Organization}, 93:258--265.

\bibitem[\protect\astroncite{Capraro}{2017}]{capraro2017does}
Capraro, V. (2017).
\newblock Does the truth come naturally? {T}ime pressure increases honesty in
  one-shot deception games.
\newblock {\em Economics Letters}, 158:54--57.

\bibitem[\protect\astroncite{Capraro}{2018}]{capraro2018gender}
Capraro, V. (2018).
\newblock Gender differences in lying in sender-receiver games: A
  meta-analysis.
\newblock {\em Judgment and Decision Making}, 13(4):345--355.

\bibitem[\protect\astroncite{Capraro and Cococcioni}{2015}]{capraro2015social}
Capraro, V. and Cococcioni, G. (2015).
\newblock Social setting, intuition and experience in laboratory experiments
  interact to shape cooperative decision-making.
\newblock {\em Proceedings of the Royal Society B: Biological Sciences},
  282(1811):20150237.

\bibitem[\protect\astroncite{Capraro and
  Cococcioni}{2016}]{capraro2016rethinking}
Capraro, V. and Cococcioni, G. (2016).
\newblock Rethinking spontaneous giving: Extreme time pressure and
  ego-depletion favor self-regarding reactions.
\newblock {\em Scientific Reports}, 6:27219.

\bibitem[\protect\astroncite{Capraro et~al.}{2017}]{capraro2017deliberation}
Capraro, V., Corgnet, B., Esp{\'\i}n, A.~M., and Hern{\'a}n-Gonz{\'a}lez, R.
  (2017).
\newblock Deliberation favours social efficiency by making people disregard
  their relative shares: {E}vidence from {USA} and {I}ndia.
\newblock {\em Royal Society Open Science}, 4(2):160605.

\bibitem[\protect\astroncite{Capraro and Kuilder}{2016}]{capraro2016know}
Capraro, V. and Kuilder, J. (2016).
\newblock To know or not to know? {L}ooking at payoffs signals selfish
  behavior, but it does not actually mean so.
\newblock {\em Journal of Behavioral and Experimental Economics}, 65:79--84.

\bibitem[\protect\astroncite{Capraro et~al.}{2019}]{capraro2019time}
Capraro, V., Schulz, J., and Rand, D.~G. (2019).
\newblock Time pressure and honesty in a deception game.
\newblock {\em Journal of Behavioral and Experimental Economics}, (73):93--99.

\bibitem[\protect\astroncite{Childs}{2012}]{childs2012gender}
Childs, J. (2012).
\newblock Gender differences in lying.
\newblock {\em Economics Letters}, 114(2):147--149.

\bibitem[\protect\astroncite{Corgnet et~al.}{2015}]{corgnet2015cognitive}
Corgnet, B., Esp{\'\i}n, A.~M., and Hern{\'a}n-Gonz{\'a}lez, R. (2015).
\newblock The cognitive basis of social behavior: cognitive reflection
  overrides antisocial but not always prosocial motives.
\newblock {\em Frontiers in behavioral Neuroscience}, 9:287.

\bibitem[\protect\astroncite{Debey et~al.}{2012}]{debey2012lying}
Debey, E., Verschuere, B., and Crombez, G. (2012).
\newblock Lying and executive control: An experimental investigation using ego
  depletion and goal neglect.
\newblock {\em Acta Psychologica}, 140(2):133--141.

\bibitem[\protect\astroncite{Dreber and Johannesson}{2008}]{dreber2008gender}
Dreber, A. and Johannesson, M. (2008).
\newblock Gender differences in deception.
\newblock {\em Economics Letters}, 99(1):197--199.

\bibitem[\protect\astroncite{Erat and Gneezy}{2012}]{erat2012white}
Erat, S. and Gneezy, U. (2012).
\newblock White lies.
\newblock {\em Management Science}, 58(4):723--733.

\bibitem[\protect\astroncite{Ezquerra et~al.}{2018}]{ezquerra2018gender}
Ezquerra, L., Kolev, G.~I., and Rodriguez-Lara, I. (2018).
\newblock Gender differences in cheating: Loss vs. gain framing.
\newblock {\em Economics Letters}, 163:46--49.

\bibitem[\protect\astroncite{Fischbacher and
  F{\"o}llmi-Heusi}{2013}]{fischbacher2013lies}
Fischbacher, U. and F{\"o}llmi-Heusi, F. (2013).
\newblock Lies in disguise—an experimental study on cheating.
\newblock {\em Journal of the European Economic Association}, 11(3):525--547.

\bibitem[\protect\astroncite{Friesen and
  Gangadharan}{2012}]{friesen2012individual}
Friesen, L. and Gangadharan, L. (2012).
\newblock Individual level evidence of dishonesty and the gender effect.
\newblock {\em Economics Letters}, 117(3):624--626.

\bibitem[\protect\astroncite{Gino et~al.}{2013}]{gino2013self}
Gino, F., Ayal, S., and Ariely, D. (2013).
\newblock Self-serving altruism? {T}he lure of unethical actions that benefit
  others.
\newblock {\em Journal of Economic Behavior \& Organization}, 93:285--292.

\bibitem[\protect\astroncite{Gino et~al.}{2011}]{gino2011unable}
Gino, F., Schweitzer, M.~E., Mead, N.~L., and Ariely, D. (2011).
\newblock Unable to resist temptation: How self-control depletion promotes
  unethical behavior.
\newblock {\em Organizational Behavior and Human Decision Processes},
  115(2):191--203.

\bibitem[\protect\astroncite{Gneezy}{2005}]{gneezy2005deception}
Gneezy, U. (2005).
\newblock Deception: The role of consequences.
\newblock {\em American Economic Review}, 95(1):384--394.

\bibitem[\protect\astroncite{Gneezy et~al.}{2018}]{gneezy2018lying}
Gneezy, U., Kajackaite, A., and Sobel, J. (2018).
\newblock Lying aversion and the size of the pie.
\newblock {\em American Economic Review}, 108:419--453.

\bibitem[\protect\astroncite{Goodman et~al.}{2013}]{goodman2013data}
Goodman, J.~K., Cryder, C.~E., and Cheema, A. (2013).
\newblock Data collection in a flat world: The strengths and weaknesses of
  mechanical turk samples.
\newblock {\em Journal of Behavioral Decision Making}, 26(3):213--224.

\bibitem[\protect\astroncite{Gravelle}{2009}]{gravelle2009tax}
Gravelle, J.~G. (2009).
\newblock Tax havens: International tax avoidance and evasion.
\newblock {\em National Tax Journal}, pages 727--753.

\bibitem[\protect\astroncite{Gunia et~al.}{2012}]{gunia2012contemplation}
Gunia, B.~C., Wang, L., Huang, J., and Murninghan, J.~K. (2012).
\newblock Contemplation and conversation: Subtle influences of moral decision
  making.
\newblock {\em Academy of Management Journal}, 55:13--33.

\bibitem[\protect\astroncite{Horton et~al.}{2011}]{horton2011online}
Horton, J.~J., Rand, D.~G., and Zeckhauser, R.~J. (2011).
\newblock The online laboratory: Conducting experiments in a real labor market.
\newblock {\em Experimental Economics}, 14(3):399--425.

\bibitem[\protect\astroncite{Hurkens and Kartik}{2009}]{hurkens2009would}
Hurkens, S. and Kartik, N. (2009).
\newblock Would {I} lie to you? {O}n social preferences and lying aversion.
\newblock {\em Experimental Economics}, 12(2):180--192.

\bibitem[\protect\astroncite{Jordan et~al.}{2016}]{jordan2016uncalculating}
Jordan, J.~J., Hoffman, M., Nowak, M.~A., and Rand, D.~G. (2016).
\newblock Uncalculating cooperation is used to signal trustworthiness.
\newblock {\em Proceedings of the National Academy of Sciences},
  113(31):8658--8663.

\bibitem[\protect\astroncite{Kajackaite and
  Gneezy}{2017}]{kajackaite2017incentives}
Kajackaite, A. and Gneezy, U. (2017).
\newblock Incentives and cheating.
\newblock {\em Games and Economic Behavior}, 102:433--444.

\bibitem[\protect\astroncite{Kartik}{2009}]{kartik2009strategic}
Kartik, N. (2009).
\newblock Strategic communication with lying costs.
\newblock {\em The Review of Economic Studies}, 76(4):1359--1395.

\bibitem[\protect\astroncite{Levine and Schweitzer}{2014}]{levine2014liars}
Levine, E.~E. and Schweitzer, M.~E. (2014).
\newblock Are liars ethical? {O}n the tension between benevolence and honesty.
\newblock {\em Journal of Experimental Social Psychology}, 53:107--117.

\bibitem[\protect\astroncite{Levine and Schweitzer}{2015}]{levine2015prosocial}
Levine, E.~E. and Schweitzer, M.~E. (2015).
\newblock Prosocial lies: When deception breeds trust.
\newblock {\em Organizational Behavior and Human Decision Processes},
  126:88--106.

\bibitem[\protect\astroncite{Lohse et~al.}{2018}]{lohse2018deception}
Lohse, T., Simon, S.~A., and Konrad, K.~A. (2018).
\newblock Deception under time pressure: Conscious decision or a problem of
  awareness?
\newblock {\em Journal of Economic Behavior \& Organization}, 146:31--42.

\bibitem[\protect\astroncite{Lotz}{2015}]{lotz2015spontaneous}
Lotz, S. (2015).
\newblock Spontaneous giving under structural inequality: Intuition promotes
  cooperation in asymmetric social dilemmas.
\newblock {\em PLoS One}, 10(7):e0131562.

\bibitem[\protect\astroncite{Mason and Suri}{2012}]{mason2012conducting}
Mason, W. and Suri, S. (2012).
\newblock Conducting behavioral research on {A}mazon’s {M}echanical {T}urk.
\newblock {\em Behavior Research Methods}, 44(1):1--23.

\bibitem[\protect\astroncite{Mazar et~al.}{2008}]{mazar2008dishonesty}
Mazar, N., Amir, O., and Ariely, D. (2008).
\newblock The dishonesty of honest people: A theory of self-concept
  maintenance.
\newblock {\em Journal of Marketing Research}, 45(6):633--644.

\bibitem[\protect\astroncite{Mead et~al.}{2009}]{mead2009too}
Mead, N.~L., Baumeister, R.~F., Gino, F., Schweitzer, M.~E., and Ariely, D.
  (2009).
\newblock Too tired to tell the truth: Self-control resource depletion and
  dishonesty.
\newblock {\em Journal of Experimental Social Psychology}, 45(3):594--597.

\bibitem[\protect\astroncite{{Open Science
  Collaboration}}{2015}]{open2015estimating}
{Open Science Collaboration} (2015).
\newblock Estimating the reproducibility of psychological science.
\newblock {\em Science}, 349(6251):aac4716.

\bibitem[\protect\astroncite{Paolacci and Chandler}{2014}]{paolacci2014inside}
Paolacci, G. and Chandler, J. (2014).
\newblock Inside the turk: Understanding mechanical turk as a participant pool.
\newblock {\em Current Directions in Psychological Science}, 23(3):184--188.

\bibitem[\protect\astroncite{Paolacci et~al.}{2010}]{paolacci2010running}
Paolacci, G., Chandler, J., and Ipeirotis, P.~G. (2010).
\newblock Running experiments on amazon mechanical turk.
\newblock {\em Judgment and Decision Making}, 5:411--419.

\bibitem[\protect\astroncite{Rand}{2016}]{rand2016cooperation}
Rand, D.~G. (2016).
\newblock Cooperation, fast and slow: Meta-analytic evidence for a theory of
  social heuristics and self-interested deliberation.
\newblock {\em Psychological Science}, 27(9):1192--1206.

\bibitem[\protect\astroncite{Rand et~al.}{2016}]{rand2016social}
Rand, D.~G., Brescoll, V.~L., Everett, J.~A., Capraro, V., and Barcelo, H.
  (2016).
\newblock Social heuristics and social roles: Intuition favors altruism for
  women but not for men.
\newblock {\em Journal of Experimental Psychology: General}, 145(4):389--396.

\bibitem[\protect\astroncite{Rand et~al.}{2012}]{rand2012spontaneous}
Rand, D.~G., Greene, J.~D., and Nowak, M.~A. (2012).
\newblock Spontaneous giving and calculated greed.
\newblock {\em Nature}, 489(7416):427--430.

\bibitem[\protect\astroncite{Rand et~al.}{2014}]{rand2014social}
Rand, D.~G., Peysakhovich, A., Kraft-Todd, G.~T., Newman, G.~E., Wurzbacher,
  O., Nowak, M.~A., and Greene, J.~D. (2014).
\newblock Social heuristics shape intuitive cooperation.
\newblock {\em Nature Communications}, 5:3677.

\bibitem[\protect\astroncite{Shalvi and De~Dreu}{2014}]{shalvi2014oxytocin}
Shalvi, S. and De~Dreu, C.~K. (2014).
\newblock Oxytocin promotes group-serving dishonesty.
\newblock {\em Proceedings of the National Academy of Sciences}, page
  201400724.

\bibitem[\protect\astroncite{Shalvi et~al.}{2012}]{shalvi2012honesty}
Shalvi, S., Eldar, O., and Bereby-Meyer, Y. (2012).
\newblock Honesty requires time (and lack of justifications).
\newblock {\em Psychological Science}, 23(10):1264--1270.

\bibitem[\protect\astroncite{Sheremeta and Shields}{2013}]{sheremeta2013liars}
Sheremeta, R.~M. and Shields, T.~W. (2013).
\newblock Do liars believe? {B}eliefs and other-regarding preferences in
  sender--receiver games.
\newblock {\em Journal of Economic Behavior \& Organization}, 94:268--277.

\bibitem[\protect\astroncite{Spence et~al.}{2001}]{spence2001behavioural}
Spence, S.~A., Farrow, T.~F., Herford, A.~E., Wilkinson, I.~D., Zheng, Y., and
  Woodruff, P.~W. (2001).
\newblock Behavioural and functional anatomical correlates of deception in
  humans.
\newblock {\em Neuroreport}, 12(13):2849--2853.

\bibitem[\protect\astroncite{Van't~Veer et~al.}{2014}]{van2014limited}
Van't~Veer, A., Stel, M., and van Beest, I. (2014).
\newblock Limited capacity to lie: Cognitive load interferes with being
  dishonest.
\newblock {\em Judgment and Decision Making}, 9:199--206.

\bibitem[\protect\astroncite{Walczyk et~al.}{2003}]{walczyk2003cognitive}
Walczyk, J.~J., Roper, K.~S., Seemann, E., and Humphrey, A.~M. (2003).
\newblock Cognitive mechanisms underlying lying to questions: Response time as
  a cue to deception.
\newblock {\em Applied Cognitive Psychology}, 17(7):755--774.

\bibitem[\protect\astroncite{Weisel and Shalvi}{2015}]{weisel2015collaborative}
Weisel, O. and Shalvi, S. (2015).
\newblock The collaborative roots of corruption.
\newblock {\em Proceedings of the National Academy of Sciences},
  112(34):10651--10656.

\bibitem[\protect\astroncite{Wiltermuth}{2011}]{wiltermuth2011cheating}
Wiltermuth, S.~S. (2011).
\newblock Cheating more when the spoils are split.
\newblock {\em Organizational Behavior and Human Decision Processes},
  115(2):157--168.

\end{thebibliography}

\pagebreak

\appendix

\section{Experimental instructions}

We report only the instructions of Study 1. Those of Study 2 were the same, a part from the differences in the list and in the self-report question discussed in Section \ref{se:study2}.

\emph{Screen 1}

Welcome to this HIT.

This is an \textbf{anonymous} HIT. It will take about ten minutes. For your participation, you will earn 50c and some additional money to be determined later.

\bigskip

\emph{Screen 2}

In the next screens, you will generate two pieces of information:
\begin{itemize}
    \item A \textbf{POSITION}, which will be a number, for example 14
    \item A \textbf{LIST} of potential bonuses between 1 and 90 cents, such as:
    $$
    23\,\,\,\,73\,\,\,\,34\,\,\,\,22\,\,\,\,2\,\,\,\,11\,\,\,\,54\,\,\,\,21\,\,\,\, 44\,\,\,\,3\,\,\,\,22\,\,\,\,6\,\,\,\,89\,\,\,\,45\,\,\,\, 67\,\,\,\, 23\,\,\,\, 65\,\,\,\,46\,\,\,\, 77\,\,\,\, 1\,\,\,\, 86\,\,\,\, 5\,\,\,\, 4
    $$
\end{itemize}

Your bonus for this HIT will be the one in the position you generated. Thus, in this case, it will be 45c, because 45 is in position 14.

 What is your job? 
 
 REPORT the POSITION you generated.

\underline{\textbf{IMPORTANT}: This survey is anonymous. We do not know the position that you will generate.}

 We now ask some comprehension questions in order to ascertain that you
 understood the situation.
 
 Assume that you generate:
 
POSITION: 13

LIST: $35\,\,\,\,24\,\,\,\,36\,\,\,\,82\,\,\,\,14\,\,\,\,3\,\,\,\,29\,\,\,\,21\,\,\,\, 55\,\,\,\, 64\,\,\,\,35\,\,\,\,33\,\,\,\,67\,\,\,\,2\,\,\,\, 13\,\,\,\, 21\,\,\,\, 65\,\,\,\, 46\,\,\,\, 8\,\,\,\, 43\,\,\,\, 32$

What \textbf{position} did you generate? (blank text box where to type the answer below the question)

What would be your bonus if you report the \textbf{position} you generated? (blank text box where to type the answer below the question)

What \textbf{position} could you report if you wanted to earn a bigger bonus? (blank text box where to type the answer below the question)

\bigskip

\emph{Screen 3}

Assume that you generate:
 
POSITION: 9

LIST: $53\,\,\,\,8\,\,\,\,24\,\,\,\,27\,\,\,\,76\,\,\,\,61\,\,\,\,65\,\,\,\,14\,\,\,\, 34\,\,\,\, 25\,\,\,\,22\,\,\,\,33\,\,\,\,73\,\,\,\,68\,\,\,\, 86\,\,\,\, 2\,\,\,\, 13\,\,\,\, 25\,\,\,\, 71\,\,\,\, 78\,\,\,\, 63\,\,\,\,27\,\,\,\,51\,\,\,\,11$

What \textbf{position} did you generate? (blank text box where to type the answer below the question)

What would be your bonus if you report the \textbf{position} you generated? (blank text box where to type the answer below the question)

What \textbf{position} could you report if you wanted to earn a bigger bonus? (blank text box where to type the answer below the question)

\bigskip

\emph{Screen 4}

Congratulations, you have passed all comprehension questions.

In the next screen, you will generate a POSITION. In the screen after, you will generate a SEQUENCE of potential bonuses and you will be asked to report the POSITION you generated.

Remember that you will receive a payment in cents equal to the number in the position \underline{\textbf{you report}}.

\bigskip

\emph{Screen 5 (Unlucky condition)}

POSITION: 19

Please take note of this position on a piece of paper.

Now click the next button in order to generate a sequence.

\bigskip

\emph{Screen 5 (Lucky condition)}

POSITION: 22

Please take note of this position on a piece of paper.

Now click the next button in order to generate a sequence.

\bigskip

\emph{Screen 6}

LIST: 

$25\,\,\,\, 3\,\,\,\, 63\,\,\,\, 54\,\,\,\, 28\,\,\,\, 70\,\,\,\, 37\,\,\,\, 36\,\,\,\, 26\,\,\,\, 31\,\,\,\, 43\,\,\,\, 15\,\,\,\, 30\,\,\,\, 60\,\,\,\, 33\,\,\,\, 37\,\,\,\, 15\,\,\,\, 63\,\,\,\, 16\,\,\,\, 50\,\,\,\, 4\,\,\,\, 71\,\,\,\, 79\,\,\,\, 2\,\,\,\, 85\,\,\,\, 48$

What POSITION did you generate?

(blank text box where to type the answer here)

(The Time Pressure condition differed from this condition only in that right after the list and right before asking participants ``What POSITION did you generate?'', we added the sentence: REPLY WITHIN 15 SECONDS. OTHERWISE YOU WON'T GET ANY BONUS''.)

\section{Analysis of the Time Pressure condition}

Here we analyze the effect of time pressure on honesty. A manipulation check confirms that subjects under time pressure took much shorter to make a decision than subjects who decided with no time constraint (13.8s vs 24.4s, ranksum: p $< .001$). 

Overall, we find that the rate of honesty under time pressure is very similar to the rate of honesty without time constraint (82.3\% vs 84.1\%). Logit regression predicting \emph{Honesty} (dummy variable) as a function of \emph{Condition} (Time Pressure vs No Time Constraint) confirms that the rate of honesty in the two conditions are not significantly different (without control on sex, age, education: coeff $= -0.123$, z $= -0.64$, p $= 0.524$; with control: coeff $= -0.156$, z $= -0.78$, p $= 0.437$). This remains true when we split the sample in the Lucky versus Unlucky condition. In the Lucky condition the rates of honesty in the Time Pressure and in the No Time Constraint conditions were, respectively, 85.9\% and 91.4\% (logit regression without control: coeff $= 0.567$, z $= -1.60$, p $= 0.109$; with control: coeff $= -0.560$, z $= -1.58$, p $= 0.115$). In the Unlucky condition, the rates od honesty in the Time Pressure and in the No Time Constraint conditions were, respectively, 78.9\% and 77.6\% (logit regression without control: coeff $= 0.076$, z $= 0.30$, p $= 0.765$; with control: coeff $= 0.063$, z $= 0.24$, p $= 0.808$).

Therefore, time pressure has no effect on the overall rate of honesty. However, although time pressure has no effect on the rate of honesty, we find that it has the effect of changing the distribution of positions (Kolmogorov-Smirnov, p $= 0.038$). Where does this effect come from, if the overall rate of honesty is constant across time manipulation conditions? Table 1 summarizes the distribution of choices in the Time Pressure and in the No Time Constraint condition. The number of honest choices is essentially the same in the Time Pressure and in the No Time Constraint (294 vs 291), so is the same the number of people stretching the truth (5 people in each condition). The only difference between the Time Pressure condition and the No Time Constraint condition is given by: (i) many of the people who would choose the global maximum in the No Time Constraint condition end up choosing an early local maximum in the Time Pressure condition (arguably because they have no time to find out all the payoffs up to the global maximum); (ii) twelve participants are not classifiable (arguably because some subjects get confused by the time pressure). 

\begin{table}[h]
  \centering
  \caption{Distributions of the reported positions in the Time Pressure and in the No Time Constraint conditions of Study 1.}
    \begin{tabular}{| l | c | c |}
    \hline
     & Time Pressure & No Time Constraint \\
    \hline
    \# honest & 294 & 291 \\
    \hline
    \# choosing global max & 12 & 43 \\
\hline
    \# choosing local max & 34 & 5 \\
    \hline
    \# stretching the truth & 5 & 5 \\
    \hline
    \# not classifiable in one of the above & 12 & 0 \\
    \hline
    \end{tabular}
\end{table}

\section{Proofs}

\begin{prop}
{\rm Let $p$ be a participant with $\tau_1(p)=\ldots=\tau_n(p)=0$. Then:
\begin{itemize}
\item If $\lambda(p) > M$, that is, if the intrinsic cost of lying is greater than the maximal payoff, $p$ maximizes their utility by finding out a random number of payoffs and by reporting the true state of the world, regardless of the payoffs that $p$ has found out;
\item If $\lambda(p)\leq M$, $p$ maximizes their utility by finding out all payoffs (note that this includes $\overline{M}$), and then:
\begin{itemize}
    \item reporting the truth, if $m(\pi_g) > \overline M - \lambda(p)$;
    \item reporting a state of the world $\pi$ such that $m(\pi)=\overline M$, if $m(\pi_g) < \overline M - \lambda(p)$;
    \item reporting a state of the world chosen at random between $\pi_g$ and $\pi$ such that $m(\pi)=\overline M$, if $m(\pi_g) = \overline M - \lambda(p)$.
\end{itemize}
\end{itemize}
}
\end{prop}

\begin{proof}
Since $p$ pays no cost to find out new payoffs, then for $p$ it is optimal to find out all payoffs\footnote{There are some extreme cases in which $p$ has also other optimal strategies, but they also lead $p$ to either tell the truth or lie maximally. More precisely, if $\lambda(p)$ is large enough, then $p$ is indifferent between finding out a new payoff and reporting the truth; if one of the payoffs is equal to $M$, the maximal potential value, then $p$ is indifferent between reporting the state of the world corresponding to $M$ and finding out a new payoff. In any case, $p$ either reports the truth or lie maximally.}. Now, after finding out all payoffs, if $p$ reports the true state of the world, $\pi_g$, then $p$ gets a utility of $m(\pi_g)$. If $p$ reports a state of the world $\overline\pi$ such that $m(\overline\pi)=\overline M$, then $p$ gets a utility $U(\pi)=\overline M-\lambda$, if $\overline\pi\neq\pi_g$, and $U(\pi)=\overline M$, if $\overline\pi=\pi_g$. If $p$ reports a state $\pi'\neq\overline\pi,\pi_g$, then $p$ gets a utility $U(\pi') < U(\overline\pi)$. Thus, $p$ will either tell the truth or lie maximally. No other state of the world maximize $p$'s utility. Moreover, if $m(\pi_g) > \overline M - \lambda(p)$, then $p$ reports the truth; if $m(\pi_g) < \overline M - \lambda(p)$, then $p$ lies maximally, by reporting $\overline\pi$; if $m(\pi_g) = \overline M - \lambda(p)$, then $p$ is indifferent between reporting the truth and lying maximally.
\end{proof}

\begin{prop}
{\rm Suppose that $\lambda(p)=0$. The strategy that maximizes $p$'s utility is to first find out all payoffs up to $t^{last}$ and then report a state of the world according to the following three cases:
\begin{enumerate}
    \item if $t^{last}<n$ and $m(\pi_i)<\frac{M+1}{2}$, for all $i\leq t^{last}$, then $p$ maximizes their utility by reporting a state of the world, $\pi_r$, chosen at random among those for which $r>t^{last}$; 
    \item if $t^{last}<n$ and there exists a $\pi_i$, with $i\leq t^{last}$, such that $m(\pi_i)\geq\frac{M+1}{2}$, then  player $p$ maximizes their utility by reporting a state of the world $\pi_r$ such that $m(\pi_r)=\max_{i=1,\ldots,t^{last}}m(\pi_i) = M_{t^{last}}$.
    \item if $t^{last}=n$ then  player $p$ maximizes their utility by reporting a state of the world, $\pi_r$, such that $m(\pi_r)=\max_{i=1,\ldots,n}m(\pi_i)=\overline M$.
\end{enumerate}
}
\end{prop}

\begin{proof}

It is clear that $p$ finds out all the payoffs until time $t^{last}$ and reports a state of the world at time $t^{last}$. Which state of the world, $\pi_r$, will $p$ report will depend on the payoffs found out up to time $t^{last}$, and on $t^{last}$ itself. Specifically, we distinguish three cases:
\begin{enumerate}
    \item if $t^{last}<n$ and $m(\pi_i)<(M+1)/2$, for all $i\leq t^{last}$, then choosing to report a state of the world $\pi_r$ at random among those whose associated payoff has not been found out yet ($r>t^{last}$), will give $p$ an expected utility of $(M+1)/2$, which is strictly larger that any $m(\pi_i)<(M+1)/2$, with $i\leq t^{last}$. Thus, the choice that maximizes $p$'s utility is to choose to report at random a state of the world $\pi_r$, with $r>t^{last}$.
    \item If $t^{last}<n$ and there exists a $\pi_i$, with $i\leq t^{last}$, such that $m(\pi_i)\geq(M+1)/2$, then choosing to report a state of the world $\pi_r$ at random among those whose associated payoff has not been found out yet ($r>t^{last}$), will give $p$ an expected utility of $(M+1)/2$, which is smaller or equal than the utility that $p$ would get by reporting $\pi_r$ such that $m(\pi_r)=\max_{i=1,\ldots,t^{last}}m(\pi_i) = M_{t^{last}}$.
    \item If $t^{last}=n$, then $p$ has found out all payoffs and then $p$ maximizes their utility by reporting a state of the world corresponding to the maximum payoff $\overline{M}$.
\end{enumerate}

\end{proof}
 
 \begin{prop}\label{prop:time_zero}
{\rm It is optimal for participant $p$ to report the true state of the world $\pi_g$ at time $t_0=0$ if one of the following conditions is satisfied:
\begin{itemize}
    \item $\lambda\geq\frac{M+1}{2}$ and $\tau_1>0$;
    \item $g=1$, $\lambda<\frac{M+1}{2}$ and $\tau_1>\frac{1}{M}\sum_{k=\lfloor\frac{M+1}{2}\rfloor-\lambda+1}^M \left(k-\frac{M+1}{2}+\lambda\right)$;
    \item $g>1$, $\lambda<\frac{M+1}{2}$ and $\tau_{1}>\frac{1}{M}\sum_{k=\lfloor\frac{M+1}{2}\rfloor+\lambda+1}^M\left(k-\lambda-\frac{M+1}{2}\right)$.
\end{itemize}
}
\end{prop}

 \begin{proof}
Suppose that $p$ decides to report a state of the world at time $t_0=0$. Then $p$ obtains their maximal utility by reporting the true position $\pi_g$, which gives an expected utility of $\frac{M+1}{2}$, which is greater than the expected utility of lying, which, at time $t_0=0$, is $\frac{M+1}{2}-\lambda$.

Now, the expected utility of finding out $m(\pi_1$) and deciding at time $t_1=1$ will clearly depend on whether $g=1$ or not (which is a piece of information known to $p$). So, we distinguish the following two sub-cases.

\textbf{Case $g=1$}

If $g=1$, the utility corresponding to the optimal strategy at time $t_1$ is equal to $\max(\frac{M+1}{2}-\lambda-\tau_1,m(\pi_1)-\tau_1)$. We first consider the case in which this $max$ is equal to $\frac{M+1}{2}-\lambda-\tau_1$. This happens whenever $m(\pi_1)\leq\frac{M+1}{2}-\lambda$, which happens with probability $\frac{1}{M}\max\left(\lfloor\frac{M+1}{2}-\lambda\rfloor,0\right)$, where $\lfloor x\rfloor$ is the floor of $x$, that is, the largest integer that is smaller than or equal to $x$. Thus the expected payoff at time $t_0=0$ of finding out $m(\pi_1)$ at time $t_1=1$ is:
$$
\frac{1}{M}\max\left(\lfloor\frac{M+1}{2}\rfloor-\lambda,0\right)\left(\frac{M+1}{2}-\lambda-\tau_1\right)+ \frac{1}{M}\sum_{k=\max\left(\lfloor\frac{M+1}{2}\rfloor-\lambda,0\right)+1}^M(k-\tau_1)
$$

Now:

\begin{itemize}
    \item If $\frac{M+1}{2}-\lambda\leq0$, then this utility is equal to $\frac{1}{M}\sum_{k=1}^M(k-\tau_1) = \frac{M+1}{2}-\tau_1$, which is clearly strictly smaller than the best payoff of deciding at time $t_0=0$, which is $\frac{M+1}{2}$, because $\tau_1>0$. Thus, in this case, the best strategy of player $p$ is not to find out any payoff and report the truth at time $t_0=0$.
\end{itemize}

\begin{itemize}
    \item If  $\frac{M+1}{2}-\lambda>0$, instead of using the formula above, we do a cost-to-benefit analysis, as follows. At time $t_0$, choosing to pay a cost $\tau_1$ to find out $m(\pi_1)$ costs, indeed, $\tau_1$. As for the benefit, finding out $m(\pi_1)$ and deciding at time $t_1=1$ is profitable only when $m(\pi_1)>\frac{M+1}{2}-\lambda$, in which case the payoff is $m(\pi_1)$ instead of $\frac{M+1}{2}-\lambda$. Thus, the benefit is
$$
\frac{1}{M}\sum_{k=\lfloor\frac{M+1}{2}\rfloor-\lambda+1}^M (k-\frac{M+1}{2}+\lambda).
$$
Thus, according to this cost-benefit analysis and using the assumption that $p$ is risk neutral, we obtain that $p$ decides to find out $m(\pi_1)$ iff:

$$
\tau_1\leq \frac{1}{M}\sum_{k=\lfloor\frac{M+1}{2}\rfloor-\lambda+1}^M \left(k-\frac{M+1}{2}+\lambda\right).
$$

\end{itemize}

\textbf{Case $g>1$}

Finding out $m(\pi_1)$ and deciding at time $t_1=1$ costs $\tau_1$ and it is a profitable deviation from choosing to report $\pi_g$ at time $t_0=0$ iff $m(\pi_1)-\lambda>\frac{M+1}{2}$, in which case $p$ gets a utility of $m(\pi_1)-\lambda$ instead of $\frac{M+1}{2}$. Thus, the cost-benefit analysis corresponds to the condition:
$$
\tau_{1}\leq\frac{1}{M}\sum_{k=\lfloor\frac{M+1}{2}\rfloor+\lambda+1}^M\left(k-\lambda-\frac{M+1}{2}\right)
$$
 
 Note that if $\lambda\geq\frac{M+1}{2}$ the previous inequality is never satisfied, because the sum in the right hand side is empty. This extreme case is thus identical to the one for $g=1$.
 
\end{proof}

  \begin{prop}
  {\rm Suppose that $p$ has found out all payoffs up to $m(\pi_t)$, then:
  \begin{itemize}
      \item $p$ maximizes his expected utility by reporting, at time $t$, the true state of the world, if one of the following conditions is satisfied:
      \begin{itemize}
          \item $g\leq t$, $m(\pi_g)> M_t-\lambda$, and $
     \tau_{t+1}>\frac{1}{M}\sum_{k=m(\pi_g)+\lambda+1}^M(k-\lambda-m(\pi_g))
     $; or 
     \item $g=t+1$, $M_t-\lambda<\frac{M+1}{2}$, and $
         \tau_{t+1}>\frac{1}{M}\sum_{k=1}^{M_t-\lambda}\left(M_t-\lambda-\frac{M+1}{2}\right)+\frac{1}{M}\sum_{k=M_t-\lambda+1}^{M}\left(k-\frac{M+1}{2}\right)
         $; or
        \item $g>t+1$, $M_t-\lambda<\frac{M+1}{2}$, and  $
         \tau_{t+1}>\frac{1}{M}\sum_{k=\lfloor\frac{M+1}{2}\rfloor+\lambda+1}^M\left(k-\lambda-\frac{M+1}{2}\right)
         $.
      \end{itemize}

      \item $p$ maximizes his expected utility by reporting, at time $t$, a state of the world corresponding to a payoff of $M_t$ (if there is more than one, then $p$ chooses at random among them), if one of the following conditions is satisfied:
      \begin{itemize}
          \item $g\leq t$, $m(\pi_g)< M_t-\lambda$, and $
     \tau_{t+1}>\frac{1}{M}\sum_{k=m(\pi_g)+\lambda+1}^M(k-\lambda-m(\pi_g))
     $; or 
     \item $g=t+1$, $M_t-\lambda>\frac{M+1}{2}$, and $
         \tau_{t+1}>\frac{1}{M}\sum_{k=1}^{M_t-\lambda}\left(M_t-\lambda-\frac{M+1}{2}\right)+\frac{1}{M}\sum_{k=M_t-\lambda+1}^{M}\left(k-\frac{M+1}{2}\right)
         $; or
        \item $g>t+1$, $M_t-\lambda>\frac{M+1}{2}$, and  $
         \tau_{t+1}>\frac{1}{M}\sum_{k=\lfloor\frac{M+1}{2}\rfloor+\lambda+1}^M\left(k-\lambda-\frac{M+1}{2}\right)
         $.
      \end{itemize}
      \item In all other situations, $p$ maximizes his expected utility by paying, at time $t$, a cost $\tau_{t+1}$ to find out $m(\pi_{t+1})$ at time $t+1$.
  \end{itemize}
  }
  \end{prop}

\begin{proof}
 Suppose that $p$ has found out all payoffs up to $m(\pi_t)$, for some $t<n$.  Let us denote $M_t=\max(m(\pi_1),m(\pi_2),\ldots,m(\pi_t),\frac{M+1}{2})$. 
 
 We distinguish three cases, depending on whether $g\leq t$, $g=t+1$ or $g>t+1$.
 
 \begin{itemize}
     \item If $g\leq t$, then $p$'s maximal utility at time $t$ is $\max(m(\pi_g)-\tau_1-\ldots-\tau_t, M_t-\lambda-\tau_1-\ldots-\tau_t)$. Thus, we distinguish two subcases:
     \begin{itemize}
     \item If $m(\pi_g)\geq M_t-\lambda$, then, at time $t$, $p$'s maximal utility is $m(\pi_g)-\tau_1-\ldots-\tau_t$, while, at time $t+1$, $p$'s maximal utility is $\max(m(\pi_g)-\tau_1-\ldots-\tau_{t+1},m(\pi_{t+1})-\lambda-\tau_1-\ldots-\tau_{t+1})$. Thus, the cost of finding out $m(\pi_{t+1})$ is smaller than the benefit iff
     $$
     \tau_{t+1}\leq\frac{1}{M}\sum_{k=m(\pi_g)+\lambda+1}^M(k-\lambda-m(\pi_g))
     $$
     Under these circumstances, it clearly follows that lucky participants, as well as those with a high cost of lying, are more likely to tell the truth at time $t$, without finding out the payoff corresponding to time $t+1$.
     \item If $m(\pi_g)< M_t-\lambda$, then, at time $t$, $p$'s maximal utility is $M_t-\lambda-\tau_1-\ldots-\tau_t$, while, at time $t+1$, $p$'s maximal utility is $\max(M_t-\lambda-\tau_1-\ldots-\tau_{t+1},m(\pi_{t+1})-\lambda-\tau_1-\ldots-\tau_{t+1})$. Thus, $p$ decide to find out $m(\pi_{t+1})$ iff
     $$
     \tau_{t+1}\leq\frac{1}{M}\sum_{k=M_t+1}^M(k-M_t).
     $$
     From this case, it clearly follows that the larger $M_t$ is, the more likely it will be for $p$ to stop. 
     \end{itemize}
     \item If $g=t+1$, then, at time $t$, $p$'s maximal utility is $\max(M_t-\lambda-\tau_1-\ldots-\tau_t, \frac{M+1}{2}-\tau_1-\ldots-\tau_t)$. Thus, we distinguish two subcases:
     \begin{itemize}
         \item if $M_t-\lambda\leq\frac{M+1}{2}$, then $p$'s maximal utility at time $t$ is $\frac{M+1}{2}-\tau_1-\ldots-\tau_t$, while, at time $t+1$, $p$'s maximal utility  is equal to $\max(m(\pi_{t+1})-\tau_1-\ldots-\tau_{t+1},M_t-\lambda-\tau_1-\ldots-\tau_{t+1})$. Therefore, finding out $m(\pi_{t+1})$ is a profitable deviation from choosing to report the truth at time $t$ iff
         $$
         \tau_{t+1}\leq\frac{1}{M}\sum_{k=1}^{M_t-\lambda}\left(M_t-\lambda-\frac{M+1}{2}\right)+\frac{1}{M}\sum_{k=M_t-\lambda+1}^{M}\left(k-\frac{M+1}{2}\right).
         $$
         \item If $M_t-\lambda>\frac{M+1}{2}$, then $p$'s maximal utility at time $t$ is $M_t-\lambda-\tau_1-\ldots-\tau_t$, while, at time $t+1$, $p$'s maximal utility is $\max(m(\pi_g)-\tau_1-\ldots-\tau_{t+1},M_t-\lambda-\tau_1-\ldots-\tau_{t+1})$. Therefore, $p$ decides to find out $m(\pi_g)$ iff
         $$
         \tau_{t+1}\leq\frac{1}{M}\sum_{k=M_t-\lambda+1}^M(k-M_t+\lambda).
         $$
     \end{itemize}
     \item If $g>t+1$, then, at time $t$, $p$'s maximal utility is $\max(M_t-\lambda-\tau_1-\ldots-\tau_t,\frac{M+1}{2}-\tau_1-\ldots-\tau_t)$. So, we have two subcases:
     \begin{itemize}
         \item If $M_t-\lambda\leq\frac{M+1}{2}$, then $p$'s maximal utility at time $t$ is $\frac{M+1}{2}-\tau_1-\ldots-\tau_t$, while, at time $t+1$, $p$'s maximal expected utility is $\max(m(\pi_{t+1})-\lambda-\tau_1-\ldots-\tau_{t+1}, \frac{M+1}{2}-\tau_1-\ldots-\tau_{t+1})$. So, $p$ decides to find out $m(\pi_{t+1})$ iff
         $$
         \tau_{t+1}\leq\frac{1}{M}\sum_{k=\lfloor\frac{M+1}{2}\rfloor+\lambda+1}^M\left(k-\lambda-\frac{M+1}{2}\right).
         $$
         \item If $M_t-\lambda>\frac{M+1}{2}$, then $p$'s maximal utility at time $t$ is $M_t-\lambda-\tau_1-\ldots-\tau_t$, while, at time $t+1$, $p$'s maximal utility is $\max(m(\pi_{t+1})-\lambda-\tau_1-\ldots-\tau_{t+1}, M_t-\lambda-\tau_1-\ldots-\tau_{t+1})$. So, $p$ decides to find out $m(\pi_{t+1})$ iff
         $$
         \tau_{t+1}\leq\frac{1}{M}\sum_{k=M_t+1}^M(k-M_t).
         $$
     \end{itemize}
 \end{itemize}
 
 \end{proof}

\end{document}